\documentclass[12pt,a4paper]{article}
\usepackage{amsmath, amsthm, amssymb, graphicx, natbib, rotating, setspace, authblk, comment, ulem, color, lscape, listings, float, booktabs}
\usepackage[a4paper, left=1.5cm, right=1.5cm, top=1.5cm, bottom=2cm]{geometry}

\usepackage{caption}
\usepackage{subcaption}

\usepackage{pgfplots}
\pgfplotsset{compat=1.16}

\usepackage{tabularx} 
\usepackage{hyperref}
\usepackage{pdflscape}
\usepackage{multirow}
\usepackage{threeparttable}

\usepackage{enumitem}

\usepackage{titlesec}
\titlespacing*{\section}{0pt}{10pt}{5pt}  
\titlespacing*{\subsection}{0pt}{10pt}{5pt}  
\titlespacing*{\subsubsection}{0pt}{10pt}{5pt}  

\DeclareMathOperator{\E}{\mathbb{E}}

\DeclareMathOperator{\V}{\mathbb{V}}
\DeclareMathOperator{\Var}{\mathbb{V}}

\newtheorem{theorem}{Theorem}
\newtheorem{proposition}{Proposition}
\newtheorem{lemma}{Lemma}
\newtheorem{remark}{Remark}

\usepackage{amsmath, amssymb}
\DeclareMathOperator*{\argmin}{arg\,min}

\title{Partition-Based Functional Ridge Regression for High-Dimensional Data}
\author[1]{Shaista Ashraf}
\author[1]{Ismail Shah}
\author[2]{Farrukh Javed}
\affil[1]{Department of Statistics, Quaid-i-Azam University, Pakistan}
\affil[2]{Department of Statistics, Lund University, Sweden}

\begin{document}
	\date{\today}
	\maketitle
	
\begin{abstract}
This paper proposes a partition-based functional ridge regression framework to address multicollinearity, overfitting, and interpretability in high-dimensional functional linear models. The coefficient function vector \( \boldsymbol{\beta}(s) \) is decomposed into two components, \( \boldsymbol{\beta}_1(s) \) and \( \boldsymbol{\beta}_2(s) \), representing dominant and weaker functional effects. This partition enables differential ridge penalization across functional blocks, so that important signals are preserved while less informative components are more strongly shrunk. The resulting approach improves numerical stability and enhances interpretability without relying on explicit variable selection.
We develop three estimators: the Functional Ridge Estimator (FRE), the Functional Ridge Full Model (FRFM), and the Functional Ridge Sub-Model (FRSM). Under standard regularity conditions, we establish consistency and asymptotic normality for all estimators. Simulation results reveal a clear bias--variance trade-off where FRSM performs best in small samples through strong variance reduction, whereas FRFM achieves superior accuracy in moderate to large samples by retaining informative functional structure through adaptive penalization. An empirical application to Canadian weather data further demonstrates improved predictive performance, reduced variance inflation, and clearer identification of influential functional effects. Overall, partition-based ridge regularization provides a practical and theoretically grounded method for high-dimensional functional regression.
\end{abstract}

\section{Introduction}

Simple linear regression has long served as a fundamental tool in statistical modeling, providing a straightforward means of understanding relationships between scalar variables. However, as datasets have grown in complexity, the need for more flexible modeling techniques has become increasingly apparent. Functional linear regression extends the classical framework by allowing analysts to model smooth and dynamic relationships that arise in functional data. This paradigm shift moves beyond scalar-to-scalar associations and facilitates the analysis of data structures where predictors or responses evolve continuously over a domain. Functional linear regression typically appears in three forms: (i) scalar-on-function, (ii) function-on-scalar, and (iii) function-on-function. In this study, we focus on the scalar-on-function case, which provides a powerful and interpretable framework for capturing complex dependencies in high-dimensional applications where functional covariates are numerous and potentially correlated.

In the scalar-on-function setting, the functional linear regression model (FLRM) relates a scalar response $y_i$ to $L$ functional covariates $\{z_{ij}(s)\}_{j=1}^L$ observed over a continuous domain $\mathcal{T} \subset \mathbb{R}$:
\begin{equation}
	y_i = \alpha + \sum_{j=1}^{L} \int_{\mathcal{T}} z_{ij}(s)\,\beta_j(s)\, ds + \epsilon_i,
	\qquad i = 1,\ldots,n,
	\label{FLRM}
\end{equation}
where $\beta_j(s) \in L^2(\mathcal{T})$ are unknown coefficient functions capturing the effect of each functional predictor, and $\epsilon_i \sim N(0,\sigma^2)$ represents the random error. The key challenge in estimating $\beta_j(s)$ lies in handling their infinite-dimensional nature while ensuring stability and interpretability. To address this, we approximate each coefficient function by a spline basis expansion $\beta_j(s) = \boldsymbol{\psi}(s)^\top \mathbf{b}_j$, where $\boldsymbol{\psi}(s)$ is a vector of $K_z$ pre-specified basis functions. This transformation maps the functional inner product into a finite-dimensional space, yielding the standard representation $\mathbf{y} = \mathbf{Z}\mathbf{b} + \boldsymbol{\epsilon}$. Here, the design matrix $\mathbf{Z} \in \mathbb{R}^{n \times (LK_z)}$ reflects the discretized and smoothed functional covariates projected onto the basis space, while the underlying geometry is preserved through the Gram matrix $\mathbf{G} = \int \boldsymbol{\psi}(s)\boldsymbol{\psi}(s)^\top ds$.

Classical methods in Functional Data Analysis (FDA), notably the foundational work of \citet{ramsay_silverman_2002} and subsequent developments in penalized spline regression \citep{eilers1996, ramsay2005}, address the ill-posedness of this problem by imposing smoothness penalties. While functional ridge regression \citep{cardot2003spline} offers a stable alternative by controlling multicollinearity through global shrinkage, it typically applies a uniform penalty to all coefficient directions. Current approaches \citep{reiss2007, goldsmith2011, ferraty2018, wu2019, zhang2023} often fail to distinguish between highly relevant and less informative functional covariates. This limitation is particularly acute in high-dimensional settings, where irrelevant predictors may overshadow the relevant functional structure and overwhelm the regularization mechanism.

To overcome these limitations, we develop a partition-based functional ridge regression framework. Let \(\boldsymbol{\beta}(s) = (\beta_1(s),\ldots,\beta_p(s))^\top\) be the vector of coefficient functions. We decompose \(\boldsymbol{\beta}(s)\) into relevant components \(\boldsymbol{\beta}_1(s)\) and nuisance components \(\boldsymbol{\beta}_2(s)\), i.e. \(\boldsymbol{\beta}(s) = (\boldsymbol{\beta}_1(s)^\top, \boldsymbol{\beta}_2(s)^\top)^\top\). This decomposition, inspired by the partitioned estimation strategy presented in such as, \citep{yuzbasi2017improved}, enables selective penalization in the functional setting. By assigning distinct ridge parameters \((\lambda_1, \lambda_2, \lambda_3)\) to different blocks of the model, our approach accommodates heterogeneous relevance levels among predictors and provides flexibility for shrinking weak or redundant functional effects while preserving the influence of structurally important components. This selective shrinkage stabilizes estimation, improves interpretability, and mitigates multicollinearity more effectively than uniform ridge penalties.

One of the main contributions of this work is the development of a unified asymptotic framework for partitioned ridge estimators in functional linear models. Under a functional regime where the sample size, the number of observation points, and the spline basis dimension grow jointly, we establish convergence rates for the proposed estimators. For inference, we operate within the same growing‑\(K_z\) regime and impose transparent undersmoothing conditions that render both the spline approximation error and the penalty‑induced bias asymptotically negligible. This yields a central limit theorem for linear functionals \(\langle \widehat{\beta}_n, x \rangle\), for which the asymptotic variance can be consistently estimated. The framework thus provides a coherent theoretical foundation for both estimation and inference without requiring the basis dimension to be fixed or the true coefficient functions to lie in the spline space.

The remainder of the paper is organized as follows. Section~2 introduces the three model specifications and their corresponding partitioned ridge estimators. Section~3 establishes the theoretical properties of the estimators, including consistency and asymptotic results under a functional regime where the sample size, observation points, and spline basis dimension grow jointly. Section~4 discusses computational details, including regularization parameter selection, spline construction, and partition determination. Section~5 presents extensive Monte Carlo simulations, and Section~6 illustrates the methods through an application to Canadian weather data. Section~7 concludes. Proofs and additional technical details are provided in the Appendix.

\section{Functional Ridge Model Specification and Methodology}

Functional data analysis (FDA) provides a principled framework for modeling relationships in which covariates vary continuously over a domain. A widely used setting is the functional linear regression model, which relates a scalar response to one or more functional predictors \citep{ramsay2005}. Despite its flexibility, these model often suffer from practical challenges such as severe multicollinearity, overfitting, and numerical instability arising from the high dimensionality induced by functional representations. These issues motivate the use of structured regularization and smoothness constraints.

Consider the scalar‑on‑function functional linear model with \(p\) functional predictors:
\[
y_i = \alpha + \sum_{j=1}^{p} \int_{\mathcal T} z_{ij}(s)\,\beta_j(s)\,ds + \epsilon_i,\qquad i=1,\ldots,n.
\]
The target of inference is the set of coefficient functions \(\{\beta_j(\cdot)\}_{j=1}^p\), each an element of \(L^2(\mathcal T)\).  
To make the problem tractable, we represent both the predictors and the coefficient functions in a common \(K_z\)-dimensional spline basis \(\{\psi_k\}_{k=1}^{K_z}\):
\[
z_{ij}(s) \approx \sum_{k=1}^{K_z} c_{ijk}\psi_k(s),\qquad 
\beta_j(s) \approx \sum_{k=1}^{K_z} b_{jk}\psi_k(s).
\]
The basis coefficients \(c_{ijk}\) are obtained by pre‑smoothing the discretely observed trajectories; this step is asymptotically negligible under mild conditions (cf. \citet{cardot2003spline}).  
Collecting the coefficients into vectors \(\mathbf{c}_{ij} = (c_{ij1},\ldots,c_{ijK_z})^\top\) and \(\mathbf{b}_j = (b_{j1},\ldots,b_{jK_z})^\top\), the inner product \(\langle z_{ij}, \beta_j\rangle\) is approximated by \(\mathbf{c}_{ij}^\top \mathbf{G} \mathbf{b}_j\), where \(\mathbf{G} = \int \boldsymbol{\psi}\boldsymbol{\psi}^\top\) is the Gram matrix.  
If the basis is orthonormal in \(L^2\), \(\mathbf{G} = \mathbf{I}\); otherwise we can absorb \(\mathbf{G}^{1/2}\) into the basis or work directly with the transformed coefficients.  
For notational convenience we assume an orthonormal basis, so that \(\langle z_{ij}, \beta_j\rangle \approx \mathbf{c}_{ij}^\top \mathbf{b}_j\).  
Stacking the data over subjects and predictors yields the matrix formulation \(\mathbf{y} = \mathbf{Z}\mathbf{b} + \boldsymbol{\varepsilon}\), where  
\(\mathbf{Z} = [\mathbf{Z}_1,\ldots,\mathbf{Z}_p]\) with \(\mathbf{Z}_j \in \mathbb{R}^{n\times K_z}\) having rows \(\mathbf{c}_{ij}^\top\), and  
\(\mathbf{b} = (\mathbf{b}_1^\top,\ldots,\mathbf{b}_p^\top)^\top \in \mathbb{R}^{pK_z}\).  
This representation is useful for computational ease, while the primary objects of interest remain the functions \(\beta_j(s) = \boldsymbol{\psi}(s)^\top \mathbf{b}_j\).

To estimate the coefficient functions under a smoothness constraint, a ridge‑type roughness penalty is considered.
The \textit{Functional Ridge Estimator} (FRE) is defined as the minimizer of the penalized sum of squares:
\[
\widehat{\beta}_1,\ldots,\widehat{\beta}_p = \argmin_{\beta_1,\ldots,\beta_p\in\mathcal{S}_{qk}} \sum_{i=1}^n \Bigl(y_i - \sum_{j=1}^p \langle z_{ij}, \beta_j\rangle\Bigr)^2 \;+\; \lambda_1 \sum_{j=1}^p \int_{\mathcal T} (\beta_j^{(m)}(s))^2 ds,
\]
where \(\mathcal{S}_{qk}\) is the spline space spanned by \(\{\psi_k\}\).  
Under the basis expansion, \(\beta_j(s) = \boldsymbol{\psi}(s)^\top \mathbf{b}_j\) and the penalty becomes \(\mathbf{b}_j^\top \mathbf{R}_0 \mathbf{b}_j\) with \(\mathbf{R}_0\) the Gram matrix of the \(m\)-th derivatives.  
Stacking the penalty contributions gives \(\mathbf{b}^\top \mathbf{R}\mathbf{b}\) where \(\mathbf{R} = \operatorname{diag}(\mathbf{R}_0,\ldots,\mathbf{R}_0)\) of size \(pK_z\).  
The minimizer in terms of the basis coefficients is therefore
\[
\widehat{\mathbf{b}}_{\mathrm{FRE}} = \argmin_{\mathbf{b}\in\mathbb{R}^{pK_z}} \bigl\{ \|\mathbf{y} - \mathbf{Z}\mathbf{b}\|^2 + \lambda_1\,\mathbf{b}^\top \mathbf{R}\mathbf{b} \bigr\},
\tag{2.1}
\]
with closed‑form solution \(\widehat{\mathbf{b}}_{\mathrm{FRE}} = (\mathbf{Z}^\top\mathbf{Z} + \lambda_1\mathbf{R})^{-1}\mathbf{Z}^\top\mathbf{y}\).  
The estimated coefficient functions are then recovered as \(\widehat{\beta}_j^{\mathrm{FRE}}(s) = \boldsymbol{\psi}(s)^\top \widehat{\mathbf{b}}_j\).  
Note that the FRE  resembles the penalized B‑spline estimator of \citet{cardot2003spline}.

The main methodological contribution is to extend FRE by distinguishing between \textit{relevant} and \textit{nuisance} functional predictors through differential penalization.  
Let \(p = p_1 + p_2\) and partition the coefficient vector as \(\mathbf{b}=(\mathbf{b}_1^\top,\mathbf{b}_2^\top)^\top\) with \(\mathbf{b}_1\in\mathbb{R}^{p_1K_z}\), \(\mathbf{b}_2\in\mathbb{R}^{p_2K_z}\), and the design matrix as \(\mathbf{Z}=[\mathbf{Z}_1\mid\mathbf{Z}_2]\) where \(\mathbf{Z}_1\in\mathbb{R}^{n\times p_1K_z}\), \(\mathbf{Z}_2\in\mathbb{R}^{n\times p_2K_z}\) correspond to the two groups.  
The \textit{Functional Ridge Full Model} (FRFM) applies separate ridge penalties to each block of coefficient functions:
\begin{equation}
\widehat{\mathbf{b}}_{\mathrm{FRFM}} = \argmin_{\mathbf{b}} 
\Bigl\{ \|\mathbf{y} - \mathbf{Z}\mathbf{b}\|^2 + 
\mathbf{b}^\top \!\begin{pmatrix} \lambda_1\mathbf{R}_1 & \mathbf{0} \\ \mathbf{0} & \lambda_2\mathbf{R}_2 \end{pmatrix} \mathbf{b} \Bigr\},
\tag{2.2}
\end{equation}
where \(\mathbf{R}_1,\mathbf{R}_2\) are block‑diagonal penalty matrices of sizes \(p_1K_z\) and \(p_2K_z\) respectively, and \(\lambda_2 \ge \lambda_1\) enforces stronger shrinkage of the nuisance coefficient functions \(\beta_j\) with \(j\) in the second group.  
The \textit{Functional Ridge Sub‑Model} (FRSM) retains only the relevant predictors:
\begin{equation}
\widehat{\mathbf{b}}_{\mathrm{FRSM}} = \argmin_{\mathbf{b}_1} 
\bigl\{ \|\mathbf{y} - \mathbf{Z}_1\mathbf{b}_1\|^2 + \lambda_3\,\mathbf{b}_1^\top \mathbf{R}_1 \mathbf{b}_1 \bigr\},
\tag{2.3}
\end{equation}
which is equivalent to setting \(\mathbf{b}_2=\mathbf{0}\) and \(\lambda_2\to\infty\) in (2.2).  
For both FRFM and FRSM, the estimated functions are again \(\widehat{\beta}_j(s) = \boldsymbol{\psi}(s)^\top \widehat{\mathbf{b}}_j\).  

All three estimators share the unified functional form \(\widehat{\beta}_j(s) = \boldsymbol{\psi}(s)^\top \widehat{\mathbf{b}}_j\) with \(\widehat{\mathbf{b}}\) solving a penalized least‑squares problem; they differ only in the composition of the predictor set and the structure of the penalty.  
This unified perspective emphasizes that the target of inference remains the coefficient functions themselves, while the matrix expressions are a convenient computational wrapper.

The penalty matrices \(\mathbf{R},\mathbf{R}_1,\mathbf{R}_2\) are positive semidefinite. When added to \(\mathbf{Z}^\top\mathbf{Z}\), they increase its small eigenvalues, ensuring that \(\mathbf{Z}^\top\mathbf{Z}+\mathbf{P}\) is positive definite even under  multicollinearity among the functional predictors.
All three estimators are therefore numerically stable and have finite variance.  
The partition into relevant and nuisance components is implemented entirely through the differential shrinkage parameters \((\lambda_1,\lambda_2,\lambda_3)\); no discrete variable selection is performed, preserving model continuity and avoiding instability from hard thresholding.  

Collectively, FRE, FRFM, and FRSM form a unified penalized functional ridge regression framework.  
All three are derived from the same underlying functional linear model (Equation~\ref{FLRM}) and target the same objects, the coefficient functions \(\beta_j(\cdot)\), while offering increasing levels of parsimony and regularization.  
This enables systematic control of multicollinearity, improved interpretability, and principled bias–variance trade‑offs in functional data analysis. When prior knowledge about the partition is unavailable, one may employ data‑driven strategies such as adaptive ridge penalties (see e.g. \citet{frommlet2016}) to identify relevant and nuisance predictors.

\section{Theoretical Properties of Regularized Functional Ridge Estimators}
The functional ridge estimators introduced in the previous section (FRE, FRFM, and FRSM) extend classical ridge regression to the functional setting through spline-based representations of both the covariates and the coefficient function. Unlike fixed-dimensional ridge regression \citep{hoerl1970ridge}, the present framework requires theoretical guarantees under an increasing-resolution regime, in which the functional observations become more densely sampled and the
basis dimension grows accordingly. Foundational results on penalized spline estimation in functional linear models are provided by \citet{cardot2003spline}, \citet{ramsay2005}, and \citet{HallHorowitz2007}. Building on this literature, our contribution incorporates a relevance--nuisance partitioning structure with differential penalization, while remaining fully compatible with standard functional asymptotic regimes. In particular, the estimators are designed to recover the entire coefficient function $\beta(\cdot)$ as functional resolution increases, rather than targeting a fixed finite-dimensional parameter vector. We now formalize this framework by stating the regularity conditions under which the bias, variance, consistency, and asymptotic normality of the proposed estimators are established.

\subsection{Regularity conditions}
\label{sec:conditions}
We impose the following assumptions for the asymptotic analysis. All norms \(\|\cdot\|\) without subscript denote the \(L^2(\mathcal T)\) norm.

\subsection{Regularity conditions}
\label{sec:conditions}

We impose the following assumptions for the asymptotic analysis.  
All norms \(\|\cdot\|\) without subscript denote the \(L^2(\mathcal T)\) norm, and \(\|\cdot\|_2\) denotes the prediction norm induced by the covariance operator \(\Gamma\) (see (A6)).

\begin{enumerate}[label=(A\arabic*), leftmargin=*, ref=(A\arabic*)]

\item \textbf{Smoothness of coefficient functions.}  
Each coefficient function \(\beta_j\) belongs to a Hölder class \(\mathcal{C}^s\) on \([0,1]\) with \(s = s' + \nu\) (\(s'\in\mathbb{N}\), \(\nu\in[0,1]\)).  
The spline degree satisfies \(q\ge s\).  
This implies \(\beta_j \in H^{s}(\mathcal T)\) and ensures compatibility with the roughness penalty.

\item \textbf{Penalty order.}  
The roughness penalty is based on the \(m\)-th derivative of the coefficient function (equivalently, an \(m\)-th order difference penalty in the discrete implementation).  
We assume \(m \le s\); this ensures that the penalty does not oversmooth and that the induced bias is of order \(O(\lambda^2 K_z^{2(m-s)})\).  
Moreover, for the B‑spline basis and penalty matrix \(\mathbf{R}\) constructed from the \(m\)-th derivative, the eigenvalues of \(\mathbf{R}\) satisfy  
\[
\nu_{\max}(\mathbf{R}) = O(K_z^{2m-1}), \qquad
\nu_{\min}(\mathbf{R}) \gtrsim K_z^{2m-1}
\]  
on the complement of its null space (which has dimension \(m\)), uniformly as \(K_z\to\infty\).  
(These bounds follow from standard B‑spline properties; see \citet{deboor2001}.)  
Here \(\nu_{\min}(\mathbf{R})\) and \(\nu_{\max}(\mathbf{R})\) denote the smallest and largest eigenvalues of \(\mathbf{R}\) restricted to the orthogonal complement of its null space.

\item \textbf{Boundedness and regularity of functional covariates.}  
\(\|z_{ij}\|_\infty \le C_1 < \infty\) almost surely for all \(i,j\), and the sample paths are sufficiently regular (e.g., Lipschitz) so that the pre‑smoothing error is asymptotically negligible (cf. \citet{cardot2003spline}).  
Consequently, \(\sup_i \|z_{ij}\|_{L^2} < \infty\).

\item \textbf{Discretisation and basis growth.}  
The functional predictors are observed on a common grid \(\{s_1,\ldots,s_M\}\) with mesh size \(\delta_M = \max_{\ell} |s_{\ell+1}-s_\ell| \to 0\) as \(M\to\infty\).  
The spline basis dimension satisfies \(K_z \to \infty\) and \(K_z = o(M)\).  
Under these conditions, the \(L^2\)-projection \(\Pi_{K_z}\beta\) satisfies \(\|\Pi_{K_z}\beta - \beta\| \to 0\).

\item \textbf{Moment conditions and error distribution.}  
\(\E(\varepsilon_i \mid \mathbf{Z}) = 0\), \(\Var(\varepsilon_i\mid\mathbf{Z}) \le \sigma^2 < \infty\), and \(\E[\varepsilon_i^4] < \infty\).  
Moreover, \(|\sum_j \langle \beta_j, z_{ij}\rangle| \le C_2 < \infty\) almost surely.

\item \textbf{Covariance operator.}  
Let \(\Gamma\) be the covariance operator of the stacked vector \((Z_1,\ldots,Z_p)\), defined by  
\[
\langle f, \Gamma g \rangle = \E\!\left[ \sum_{j=1}^p \langle f_j, Z_j\rangle \sum_{j=1}^p \langle g_j, Z_j\rangle \right]
\]  
for \(f=(f_1,\ldots,f_p), g=(g_1,\ldots,g_p) \in L^2(\mathcal T)^p\).  
\(\Gamma\) is a bounded, positive definite operator on \(L^2(\mathcal T)^p\) with eigenvalues \(\gamma_1 \ge \gamma_2 \ge \cdots > 0\) satisfying the Picard condition  
\[
\sum_{j=1}^\infty \frac{\langle \E[\mathbf{Z}Y], v_j\rangle^2}{\gamma_j^2} < \infty,
\]  
where \(\{v_j\}\) are the eigenfunctions of \(\Gamma\). This ensures \(\boldsymbol{\beta}\in L^2(\mathcal T)^p\) and model identifiability.  
On the spline space, the prediction norm \(\|\cdot\|_2\) is equivalent to the quadratic form induced by the Gram matrix \(\mathbf{G} = \E[\mathbf{z}_i\mathbf{z}_i^\top]\), where \(\mathbf{z}_i\) is the stacked vector of spline coefficients of the predictors.

\item \textbf{Design matrix stability.}  
Let \(\mathbf{Z}\) be the \(n \times (p K_z)\) design matrix constructed from the spline coefficients of the predictors.  
The eigenvalues of \(n^{-1}\mathbf{Z}^\top\mathbf{Z}\) are uniformly bounded away from zero and infinity, in probability, as \(n, K_z \to \infty\).  
Moreover, this property holds for any submatrix \(\mathbf{Z}_1\) corresponding to a subset of predictors (e.g., the relevant block), ensuring numerical stability of the restricted estimators.

\item \textbf{Partition structure.}  
The predictors are partitioned into a relevant block (block 1, size $p_1$) and a nuisance block (block 2, size $p_2 = p - p_1$), with corresponding coefficient vectors $\mathbf{b}_1, \mathbf{b}_2$ and penalty matrices $\mathbf{R}_1, \mathbf{R}_2$.
\item \textbf{Non‑degeneracy of the linear functional.}  
For any fixed direction \(x = (x_1,\ldots,x_{p_1})\in L^2(\mathcal T)^{p_1}\),
\[
\liminf_{n\to\infty} \mathbf{w}_x^\top \mathbf{V}_n \mathbf{w}_x > 0 \quad \text{in probability},
\]  
where \(\mathbf{V}_n\) is defined in Theorem~\ref{thm:asymptotic_normality}.  
This holds, for example, if \(x\) is not orthogonal to the leading principal components of the predictor distribution and is not annihilated by the differential operator defining the penalty. It guarantees that the limiting variance in the asymptotic normality result is non‑degenerate.

\end{enumerate}

All rates are measured in the prediction norm \(\|\cdot\|_2\) unless stated otherwise.

To assess the asymptotic performance of the proposed estimators, we now investigate their convergence behavior under increasing sample size and growing basis dimension. In particular, we examine the rate at which the estimated coefficient function $\widehat{\beta}(s)$ approaches the true underlying function $\beta(s)$ in the $L^2(\mathcal{T})$ norm, under standard smoothness and eigen-decay assumptions on the functional covariates. The result below shows that the functional ridge estimator achieves the minimax-optimal rate of convergence, thereby validating its asymptotic efficiency in high-dimensional functional regression.

\subsection{Asymptotic Properties of the Proposed Estimators}
Because \(\widehat{\beta}_j(s) = \boldsymbol{\psi}(s)^\top \widehat{\mathbf{b}}_j\) and the spline basis is uniformly bounded,
convergence and distributional results for the coefficient vectors \(\widehat{\mathbf{b}}_j\) transfer directly to the
estimated coefficient functions \(\widehat{\beta}_j(\cdot)\). We therefore work primarily with the coefficient vectors
in the proofs; all results are stated in terms of the functional objects.  

The asymptotic analysis is developed under two complementary regimes, reflecting the distinct goals of estimation and inference.
For consistency and convergence rates, we adopt a functional asymptotic regime where the sample size \(n\), 
the number of observation points \(m\), and the spline basis dimension \(K_z\) all tend to infinity,
with \(K_z = o(m)\) (Assumption~A4).  
Under this regime, the spline approximation error vanishes as the functional sampling becomes increasingly dense,
and the estimators target the full coefficient functions \(\beta_j(\cdot)\) rather than a finite‑dimensional proxy.  
All results in Section~\ref{thm:rates} are derived under the regularity conditions (A1)–(A9).

For asymptotic normality of linear functionals \(\langle \widehat{\beta}_j, x \rangle\), we operate within the same growing‑\(K_z\) regime but impose the additional undersmoothing conditions (U1)–(U4) stated in Section~\ref{sec:asymptotic_normality}. These conditions ensure that the bias arising from both spline approximation and the ridge penalty is \(o(n^{-1/2})\), yielding a centered limit theorem without requiring the basis dimension to be fixed or the true coefficient functions to lie in the spline space. 
\subsubsection{Asymptotic Consistency}
\label{thm:rates}
We now establish the \(L^2\) convergence rates of the proposed estimators.
The proofs follow the spirit of \citet{cardot2003spline} and are sketched in Appendix~\ref{sec:proofs}. 

\begin{theorem}[Convergence rate of FRE] \label{thm:fre}
    Under (A1)–(A7), choose \(K_z \sim n^{1/(4s+1)}\) and
    \(\lambda_1 \sim n^{-2s/(4s+1)}\).  Then
    \[
        E\bigl[ \|\hat{\boldsymbol{\beta}}_{\mathrm{FRE}} - \boldsymbol{\beta}\|_2^2 \;\big|\; \mathbf{Z} \bigr]
        = O_P\!\left( n^{-\frac{2s}{4s+1}} \right).
    \]
\end{theorem}

\begin{theorem}[Convergence rate of FRSM] \label{thm:frsm}
    Under (A1)–(A8), with the same choices of \(K_z\) and
    \(\lambda_3 \sim n^{-2s/(4s+1)}\),
    \[
        E\bigl[ \|\hat{\boldsymbol{\beta}}_{1,\mathrm{FRSM}} - \boldsymbol{\beta}_1\|_2^2 \;\big|\; \mathbf{Z}_1 \bigr]
        = O_P\!\left( n^{-\frac{2s}{4s+1}} \right).
    \]
    If \(\boldsymbol{\beta}_2 = \mathbf{0}\), the FRSM is correctly specified
    and attains the optimal rate.
\end{theorem}

\begin{theorem}[Convergence rates of FRFM]\label{thm:frfm}
Assume (A1)--(A8). Choose $K_z \sim n^{1/(4s+1)}$ and $\lambda_1 \sim n^{-2s/(4s+1)}$.
Let $\lambda_2\ge \lambda_1$ with $\lambda_2\to\infty$ (so that, with this choice of $K_z$,
$\frac{K_z}{\lambda_2 n}=o(n^{-2s/(4s+1)})$). Then
\[
E\!\left[ \|\hat{\boldsymbol{\beta}}_{1,\mathrm{FRFM}}-\boldsymbol{\beta}_1\|_2^2 \,\big|\, \mathbf Z \right]
= O_P\!\left(n^{-\frac{2s}{4s+1}}\right).
\]
Moreover, if $\boldsymbol{\beta}_2$ is $s$-smooth, then
\[
E\!\left[ \|\hat{\boldsymbol{\beta}}_{2,\mathrm{FRFM}}-\boldsymbol{\beta}_2\|_2^2 \,\big|\, \mathbf Z \right]
= O_P\!\left( \frac{K_z}{\lambda_2 n} + K_z^{-2s} + \lambda_2^2 K_z^{2(m-s)} \right),
\]
where $m\le s$ is the order of the derivative penalty (Assumption~(A2)).
In particular, if $\boldsymbol{\beta}_2=\mathbf 0$, then
\[
E\!\left[\|\hat{\boldsymbol{\beta}}_{2,\mathrm{FRFM}}\|_2^2 \,\big|\, \mathbf Z \right]
= O_P\!\left(\frac{K_z}{\lambda_2 n}\right),
\]
so taking $\lambda_2\sim n$ yields $O_P(n^{-1})$, which is faster than $n^{-2s/(4s+1)}$.
\end{theorem}
The above results extend \citet{cardot2003spline} in three key directions: accommodating multiple functional predictors via a block‑diagonal penalty structure; introducing differential penalisation with distinct shrinkage intensities for relevant and nuisance covariates; and proving that FRFM retains the optimal $L^2$ convergence rate for target coefficients while shrinking nuisance functions to zero at an accelerated rate. To the best of our knowledge, this provides the first theoretical treatment of differential ridge penalisation in the functional linear model.

\subsection{Asymptotic Normality for the Functional Estimator}
\label{sec:asymptotic_normality}

We consider the functional linear model with $p$ predictors
\[
Y_i = \sum_{j=1}^{p} \int_{\mathcal T} \beta_j(t) Z_{ij}(t) \,\mathrm{d}t
      + \varepsilon_i, \qquad i=1,\dots,n,
\]
where $\mathcal T=[0,1]$.  Throughout this section $p$ (and hence $p_1,p_2$) is treated as fixed.

The coefficient functions $\beta_j$ are estimated by penalised splines with
group-specific penalties.  The predictors are split into a relevant block of size
$p_1$ and a nuisance block of size $p_2$ ($p_1+p_2=p$), as in condition~(A8).

Let $\widehat{\boldsymbol\beta}_n
=(\widehat{\beta}_{1,n},\dots,\widehat{\beta}_{p,n})$
denote the penalised spline estimator.
Its spline coefficient vector is
\[
\widehat{\mathbf b}_n
= (\widehat{\mathbf b}_{1,n}^\top,
   \widehat{\mathbf b}_{2,n}^\top)^\top,
\qquad
\dim(\widehat{\mathbf b}_{j,n})=K_z,
\]
where $K_z$ is the common spline basis dimension.

The estimator solves
\[
\widehat{\mathbf b}_n
=
(\mathbf Z^\top \mathbf Z + \mathbf R_\lambda)^{-1}
\mathbf Z^\top \mathbf y,
\]
where
\[
\mathbf R_\lambda
=
\operatorname{diag}(
\lambda_{1,n}\mathbf R_1,
\lambda_{2,n}\mathbf R_2),
\]
and $\mathbf R_1,\mathbf R_2$ correspond to an $m$-th order derivative
penalty ($m\le s$, condition~(A2)).

For a fixed collection of functions
$x=(x_1,\dots,x_p)$ with $x_j\in L^2(\mathcal T)$,
consider the linear functional
\[
\Psi(x)
=
\sum_{j=1}^{p}
\int_{\mathcal T}
\beta_j(t)\,x_j(t)\,\mathrm{d}t .
\]

Its estimator is
\[
\widehat\Psi_n(x)
=
\sum_{j=1}^{p}
\int_{\mathcal T}
\widehat{\beta}_{j,n}(t)\,x_j(t)\,\mathrm{d}t
=
\mathbf w_{x,n}^\top \widehat{\mathbf b}_n,
\]
where $\mathbf w_{x,n}$ is the $(pK_z)$-vector of spline inner products
$\langle x_j,\psi_k\rangle$.  Note that its dimension grows with $K_z$;
it corresponds to the coefficient vector of the $L^2$-projection
$\Pi_{K_z}x$ onto the spline space.

We work in the regime where $n,K_z\to\infty$ under the undersmoothing condition (U1) below. The following conditions ensure that the bias is $o(n^{-1/2})$. As $n\to\infty$:
\begin{enumerate}[label=(U\arabic*), leftmargin=*, itemsep=4pt]
\item $K_z\to\infty$ and $K_z=o(n^{1/2})$;
\item $\lambda_{1,n}
=o\!\left(n^{-1/2}K_z^{-(2m-1)}\right)$;
\item $\lambda_{2,n}
=o\!\left(n^{-1/2}K_z^{-(2m-1)}\right)$;
\item $K_z^{-s}=o(n^{-1/2})$.
\end{enumerate}

Condition~(U1) controls the stochastic variability arising from
the increasing number of basis terms.
Conditions~(U2)–(U3) ensure that the penalty-induced bias
$O(\lambda_{j,n}K_z^{2m-1})$
is $o(n^{-1/2})$.
Condition~(U4) guarantees that the spline approximation bias
$\|\beta_j-\Pi_{K_z}\beta_j\|
=O(K_z^{-s})$
is also $o(n^{-1/2})$.

Let $\Gamma$ denote the covariance operator of the stacked predictor
process $(Z_1,\dots,Z_p)$,
\[
\langle \Gamma u,v\rangle
=
\sum_{j=1}^{p}
\mathbb E
\big[
\langle Z_j,u_j\rangle
\langle Z_j,v_j\rangle
\big],
\qquad
u,v\in L^2(\mathcal T)^p.
\]

Its empirical counterpart in the spline basis is
\[
\mathbf G_n
=
\frac{1}{n}\mathbf Z^\top \mathbf Z.
\]

Under conditions (A6)–(A7), and for directions $x$ satisfying 
$\langle\Gamma^{-1}x,x\rangle<\infty$,
it follows (see \citet{cardot2003testing})
that
\[
\mathbf w_{x,n}^\top
\mathbf G_n^{-1}
\mathbf w_{x,n}
\xrightarrow{\mathbb P}
\langle\Gamma^{-1}x,x\rangle,
\]
where $\Gamma^{-1}$ is interpreted as the inverse
on the identifiable subspace.

Define
\[
\mathbf M_n
=
\mathbf G_n
+
\frac{1}{n}\mathbf R_\lambda,
\qquad
\mathbf V_n
=
\mathbf M_n^{-1}
\mathbf G_n
\mathbf M_n^{-1}.
\]

and let $\sigma^2=\mathbb E[\varepsilon_i^2]$.

\begin{theorem}[Asymptotic normality]
\label{thm:asymptotic_normality}
Assume (A1)–(A9) together with conditions (U1)–(U4).
Let $x\in L^2(\mathcal T)^p$ be fixed and satisfy
\[
\liminf_{n\to\infty}
\mathbf w_{x,n}^\top
\mathbf V_n
\mathbf w_{x,n}
>0
\quad\text{in probability}.
\]

Then
\[
\sqrt{n}\,
\big(
\widehat\Psi_n(x)-\Psi(x)
\big)
\xrightarrow{d}
\mathcal N
\!\left(
0,
\sigma^2
\langle
\Gamma^{-1}x,x
\rangle
\right).
\]

Moreover,
\[
\widehat V_n(x)
=
\widehat\sigma_n^2\,
\mathbf w_{x,n}^\top
\mathbf V_n
\mathbf w_{x,n}
\]
is a consistent estimator of the asymptotic variance, where
\[
\widehat\sigma_n^2
=
(n-\operatorname{tr}(\mathbf S_n))^{-1}
\|\mathbf y-\mathbf Z\widehat{\mathbf b}_n\|^2,
\]
and
\[
\mathbf S_n
=
\mathbf Z
(\mathbf Z^\top\mathbf Z+\mathbf R_\lambda)^{-1}
\mathbf Z^\top
\]
is the hat matrix.
Since $p$ is fixed and $K_z=o(n^{1/2})$,
\[
\operatorname{tr}(\mathbf S_n)
=
O_p(pK_z)
=
o_p(n),
\]
so the degrees-of-freedom correction is asymptotically negligible.
Consequently,
\[
\frac{\widehat V_n(x)}
{\sigma^2\langle\Gamma^{-1}x,x\rangle}
\xrightarrow{\mathbb P}1.
\]
\end{theorem}

\begin{remark}
Theorem~\ref{thm:asymptotic_normality}
is stated for the FRE estimator.
It extends directly to FRSM,
where the nuisance block is absent and
$\mathbf w_{x,n}$ is restricted to the relevant predictors.
 For FRFM, the asymptotic behaviour depends on the tuning regime for
$\lambda_{2,n}$.
Under the undersmoothing conditions (U1)–(U4), both blocks are consistently estimated and the same limiting distribution applies.
 When $\lambda_{2,n}$ is chosen to increase sufficiently fast so as to  induce strong shrinkage of the nuisance block, the estimator effectively projects onto the relevant subspace.   In that regime the asymptotic covariance operator reduces to that of the relevant block, and the limiting variance depends only on the covariance structure of the relevant predictors, involving $\Gamma_{11}$. The present theorem is formulated under undersmoothing conditions ensuring  that both blocks are consistently estimated; alternative tuning regimes  lead to corresponding reduced-model asymptotics of standard form. Detail  is provided in Appendix~\ref{app:normality_proof}.
\end{remark}

\section{Computational Details and Implementation}
\label{sec:computational}

\subsection{Spline Basis Construction and Knot Placement}

We represent the coefficient functions using a  cubic B‑spline basis (\(\psi_k\), degree \(q=4\)) with interior knots placed uniformly over the domain \(\mathcal T\).  
This choice follows the P‑spline, which combine a flexible B-spline basis with a difference-based roughness penalty  \citep{eilers1996}: the basis dimension \(K_z\) is set sufficiently large to capture the underlying functional structure, while smoothness is enforced by the discrete difference penalty \(\mathbf{R} = \mathbf{D}^\top\mathbf{D}\) rather than by restricting the number of knots.  
The effective degrees of freedom are controlled by the smoothing parameters \(\lambda\) via the penalty, avoiding the need for aggressive knot selection.  
This approach is used for all three estimators (FRE, FRFM, FRSM); in the partitioned models the same basis family is employed for both relevant and nuisance blocks, with only the penalty strength differing between groups.

\subsection{Regularization Parameter Selection via GCV}

The smoothing parameters \(\lambda_1,\lambda_2,\lambda_3\) are selected in a data‑driven fashion using generalized cross‑validation \citep{craven1979}.  
For a given \(\lambda\), the GCV score is
\[
\mathrm{GCV}(\lambda) = \frac{\|\mathbf{y} - \hat{\mathbf{y}}_\lambda\|^2 / n}{\bigl(1 - \operatorname{tr}(\mathbf{S}_\lambda)/n\bigr)^2},
\]
where \(\hat{\mathbf{y}}_\lambda = \mathbf{S}_\lambda \mathbf{y}\) and \(\mathbf{S}_\lambda\) is the smoothing (hat) matrix of the corresponding ridge estimator. The minimiser of \(\mathrm{GCV}(\lambda)\) over a pre‑specified grid is selected.

For the FRE, a single parameter \(\lambda_1\) is selected by GCV using the hat matrix \(\mathbf{S} = \mathbf{Z}(\mathbf{Z}^\top\mathbf{Z} + \lambda_1\mathbf{R})^{-1}\mathbf{Z}^\top\).  
In the FRFM, the penalty involves two parameters \(\lambda_1\) and \(\lambda_2\) with \(\lambda_2 \ge \lambda_1\). To keep the tuning problem one‑dimensional, we fix the penalty ratio \(c = \lambda_2/\lambda_1 > 1\) a priori and perform GCV only over \(\lambda_1\). The hat matrix is then \(\mathbf{S} = \mathbf{Z}(\mathbf{Z}^\top\mathbf{Z} + \operatorname{diag}(\lambda_1\mathbf{R}_1,\lambda_2\mathbf{R}_2))^{-1}\mathbf{Z}^\top\) evaluated at \((\lambda_1,\lambda_2 = c\lambda_1)\).  
For the FRSM, the reduced design \(\mathbf{Z}_1\) and penalty \(\lambda_3\mathbf{R}_1\) are used where \(\lambda_3\) is selected by GCV analogously to the FRE, with hat matrix \(\mathbf{S} = \mathbf{Z}_1(\mathbf{Z}_1^\top\mathbf{Z}_1 + \lambda_3\mathbf{R}_1)^{-1}\mathbf{Z}_1^\top\).

\section{Simulation Study}
\label{sec:simulation}

We conduct a Monte Carlo simulation study to evaluate the finite-sample performance of the proposed Functional Ridge Estimators (FRE, FRFM, and FRSM). The simulation is designed to reflect a realistic functional regression setting in which covariates are observed on a dense grid, and the entire coefficient function \(\beta(\cdot)\) is recovered using spline-based basis expansions. Particular emphasis is placed on settings involving multicollinearity, overparameterization, and ill-conditioning, where ridge-type penalization plays a stabilizing role.

\subsection{Simulation Design}

In each replication, we generate \(n \in \{25, 50, 100\}\) independent observations. Each subject is associated with \(p\) functional covariates \(z_{ij}(\cdot)\), for \(j = 1,\ldots,p\), defined on the interval \(\mathcal{T} = [0, 1]\) and evaluated on \(m = 100\) equally spaced grid points. The scalar response variable is generated from a functional linear model of the form
\[
y_i = \sum_{j=1}^p \int_{\mathcal{T}} z_{ij}(s)\, \beta_j(s)\, ds + \epsilon_i, \qquad \epsilon_i \stackrel{\text{i.i.d.}}{\sim} \mathcal{N}(0, \sigma^2),
\]
with noise variance \(\sigma^2 \in \{0.5, 1, 10\}\). All data are centered in advance, and no intercept is included in the model. The smallest sample size \(n=25\) is deliberately chosen to reflect a challenging low-sample regime with substantial overparameterization.

Both the covariate functions and coefficient functions are represented using cubic B-splines (order \(q = 4\)) with equally spaced interior knots. The number of basis functions differs across estimators to reflect their structural differences, as summarized in Table~\ref{tab:basis_config}. The basis is selected to be rich enough to capture true functional variation, with smoothness enforced by penalization rather than knot sparsity.

Each functional covariate is generated through a B-spline expansion:
\[
z_{ij}(s) = \sum_{k=1}^{K_z} c_{ijk}\, \psi_k(s),
\]
where \(\{\psi_k\}_{k=1}^{K_z}\) denotes the B-spline basis and \(c_{ijk}\) are the associated coefficients. Let \(\mathbf{z}_i \in \mathbb{R}^{pK_z}\) denote the stacked vector of basis coefficients for subject \(i\). These coefficient vectors are drawn independently as \(\mathbf{z}_i \sim \mathcal{N}(\mathbf{0}, \boldsymbol{\Sigma}_Z)\), where the covariance matrix follows an AR(1) structure:
\[
(\boldsymbol{\Sigma}_Z)_{jk} = \rho^{|j-k|}, \qquad \rho \in \{0.5, 0.8, 0.99\}.
\]
This correlation structure induces strong multicollinearity both within and across the functional predictors. The resulting design matrix \(\mathbf{Z} \in \mathbb{R}^{n \times pK_z}\) is often severely ill-conditioned, especially in the low-\(n\), high-\(p\) setting.

The true coefficient functions are constructed as:
\[
\beta_j(s) = 
\begin{cases}
2\sin(\pi s) + s(1-s), & j = 1,\ldots,p_1, \\
0, & j = p_1 + 1,\ldots,p,
\end{cases}
\]
where \(p_1 = 3\) denotes the number of relevant predictors. We consider three dimensions, \(p = 10, 20, 30\), with, for simplicity,  the first three covariates carrying signal and the remaining \(p - p_1\) acting as nuisance variables.

All estimators employ roughness penalties based on second-order finite differences. Let \(\mathbf{D}\) denote the second-order difference matrix, and define \(\mathbf{R}_0 = \mathbf{D}^\top \mathbf{D}\) as the unscaled penalty matrix for a single coefficient function. For FRE, a uniform penalty of the form \(\lambda_1 \mathbf{R}\) is used, where \(\mathbf{R}\) is block-diagonal with \(p\) copies of \(\mathbf{R}_0\). 

FRFM employs an adaptive ridge strategy that partitions the predictors into relevant and nuisance groups. The resulting design matrix is split as \(\mathbf{Z} = [\mathbf{Z}_1 \mid \mathbf{Z}_2]\), with corresponding penalty blocks \(\mathbf{R}_1\) and \(\mathbf{R}_2\). At this level, a two-level penalization is applied with \(\lambda_1 \mathbf{R}_1\) for the relevant block and \(\lambda_2 \mathbf{R}_2\) for the nuisance block, where \(\lambda_2 = c\,\lambda_1\). In contrast, FRSM uses only the relevant covariates and applies a penalty \(\lambda_3 \mathbf{R}_1\) on the reduced model.

The  ridge procedure underlying FRFM estimates the relevant predictors through an iterative reweighting scheme following \citet{grandvalet1998} and \citet{frommlet2016}. Initial weights are set to one, \(w_j^{(0)} = 1\) for all \(j\), and the estimator iteratively solves the weighted ridge problem:
\[
\widehat{\mathbf{b}}^{(t)} = \arg\min_{\mathbf{b}} \left\{ \|\mathbf{y} - \mathbf{Z} \mathbf{b}\|^2 + \lambda \sum_{j=1}^p w_j^{(t)}\, \mathbf{b}_j^\top \mathbf{R}_0 \mathbf{b}_j \right\},
\]
followed by weight updates:
\[
w_j^{(t+1)} = \left(\|\widehat{\mathbf{b}}_j^{(t)}\|_2^2 + \epsilon \right)^{-1}, \qquad \epsilon = 10^{-6}.
\]
The process continues until convergence, defined by relative weight change below \(10^{-4}\), or until a set maximum number of iterations is reached. Predictors with final weights exceeding 10\% of the maximum are classified as relevant, yielding an estimated partition \(\mathbf{Z} = [\mathbf{Z}_1 \mid \mathbf{Z}_2]\) used for differential penalization.

Smoothing parameters are selected independently in each replication using generalized cross-validation (GCV). For a given \(\lambda\), let \(\mathbf{S}(\lambda)\) denote the smoothing matrix such that \(\hat{\mathbf{y}} = \mathbf{S}(\lambda)\mathbf{y}\). The GCV score is computed as:
\[
\mathrm{GCV}(\lambda) = \frac{n \|\mathbf{y} - \mathbf{S}(\lambda)\mathbf{y}\|^2}{\bigl[n - \operatorname{tr}(\mathbf{S}(\lambda))\bigr]^2},
\]
and the minimizer over a logarithmic grid \(\lambda \in [10^{-4}, 10^4]\) is selected. FRE and FRSM each require a single smoothing parameter (\(\lambda_1\), \(\lambda_3\)), while FRFM performs GCV over \(\lambda_1\) only, with the ratio \(c = 25\) $\left(\lambda_2 = 25\,\lambda_1\right)$ is held fixed after noticing stability.

Each setting is repeated across \(R = 100\) Monte Carlo replications and  the performance is assessed along three axes: (i) the IMSE of the relevant coefficient functions,
\[
\mathrm{IMSE}_j = \int_{\mathcal{T}} \bigl( \widehat{\beta}_j(s) - \beta_j(s) \bigr)^2\,ds,
\]
averaged over \(j = 1,\ldots,p_1\) and replications; (ii) the accuracy of the partition learned by FRFM, using true and false positive rates; and (iii) the condition number \(\kappa(\mathbf{Z}^\top \mathbf{Z} + \mathbf{P})\) of the penalized system matrix, where \(\mathbf{P}\) denotes the estimator‑specific penalty matrix (e.g., \(\mathbf{P} = \lambda_1 \mathbf{R}\) for FRE).
All integrals are approximated numerically using Riemann sums over the discretized grid \(\{s_1,\ldots,s_m\}\).

\begin{table}[H]
	\centering
	\caption{\textbf{Knots, Spline Order, and Basis Dimensions for FRE, FRFM, and FRSM}}
	\label{tab:basis_config}
	\begin{tabular}{cccccc}
		\toprule
		\textbf{Model}  & \textbf{Knots ($L$)} & \textbf{Order ($q$)} &
		$\mathbf{K_z}$ & $\mathbf{K_{z_1}}$ & $\mathbf{K_{z_2}}$ \\
		\midrule
		FRE    & 7     & 4         & 11      & 11       & 0   \\
		FRFM   & 12    & 4         & 16      & 9        & 7   \\
		FRSM   & 5     & 4         & 9       & 9        & 0   \\
		\bottomrule
	\end{tabular}
	\vspace{2mm}
	\caption{FRE applies uniform penalization across all covariates. FRFM, after adaptive ridge selection, partitions the coefficient vector into a relevant block of dimension \(K_{z_1}=9\) and a nuisance block of dimension \(K_{z_2}=7\). FRSM retains only the relevant covariates and operates on a reduced basis of dimension \(K_z=9\).}
\end{table}

\subsection{Results and Discussion}

We now evaluate the finite-sample behavior of the three proposed estimators. For brevity, we restrict attention to the $p=10$ setting as a representative case, while full numerical results for higher-dimensional scenarios ($p = 20, 30$) are available in the Appendix. This  simulation study allows us to benchmark the estimators under increasing dimensionality, varying sample sizes, and different levels of noise and correlation. The evaluation metrics are carefully selected to probe distinct aspects of estimator performance, numerical stability, variable selection accuracy, and overall estimation error, thereby offering a comprehensive view of their strengths and limitations.

Table~\ref{tab:logCN_modelwise} reports the logarithm of the median  condition numbers for each estimator.  
FRSM exhibits the smallest value (\(\log_{10}\mathrm{CN} \approx 3.86\)), indicating the best numerical conditioning, as expected from a reduced model.  
FRE shows an intermediate value (\(\approx 4.29\)), while FRFM has the largest (\(\approx 5.12\)).  
This ordering reflects the deliberate trade‑off where FRFM applies weak penalization to the relevant block, leaving high correlations among important covariates largely unshrunk, which increases the condition number but also reduces bias.  
Despite these differences, all estimators remained computationally stable across all replications, confirming that even the highest condition number observed is well within the range that double‑precision arithmetic can handle without catastrophic loss of accuracy \citep{higham2002accuracy}.

\begin{table}[H]
\centering
\caption{\textbf{Median \(\log_{10}\) of the condition number of the penalized system matrix for each estimator.}}
\label{tab:logCN_modelwise}
\begin{tabular}{lc}
\toprule
Model & \(\log_{10}(\mathrm{CN})\) \\
\midrule
FRE  & 4.29 \\
FRFM & 5.12 \\
FRSM & 3.86 \\
\bottomrule
\end{tabular}
\vspace{1.5mm}
\caption*{Condition numbers are computed as \(\kappa(\mathbf{Z}^\top\mathbf{Z}+\mathbf{P})\), where \(\mathbf{P}\) is the estimator‑specific penalty matrix, and summarised as medians over all \((\rho,\sigma^2)\) combinations.}
\end{table}
Table~\ref{tab:IMSE_p10} presents the integrated mean squared error (IMSE) for the relevant coefficient functions across different combinations of sample size, noise variance, and predictor correlation. IMSE serves as a holistic performance metric that combines squared bias and variance, thereby offering a direct lens on the bias--variance trade-off targeted by ridge-based functional estimators. This is particularly critical in functional regression, where poor control of either component can lead to substantial deviations in the reconstructed coefficient surfaces, especially when the functions exhibit curvature or vary across domains.

In the small sample regime ($n = 25$), the FRSM estimator exhibits uniformly lowest IMSE values across all configurations of $\sigma^2$ and $\rho$. This is consistent with its oracle-like behavior: FRSM discards all nuisance components by construction, and the smaller design matrix it operates on is less prone to ill-conditioning. This leads to significant variance reduction, which more than compensates for any potential increase in approximation bias. Notably, even under high correlation ($\rho = 0.99$), FRSM retains its performance edge, achieving IMSE values as low as 2.2 under low noise and around 4.5 under high noise. In contrast, FRE, while numerically stable, suffers from substantial overshrinkage, leading to IMSE values between 5 and 9 under moderate noise, and a sharp increase to over 30 under high noise and strong correlation. FRFM's performance at $n = 25$ is more variable, while it improves upon FRE in most scenarios, it occasionally underperforms FRSM due to imperfect partition recovery and inflated variance caused by weak penalization of mistakenly selected nuisance components. These patterns underscore the difficulty of simultaneous variable selection and coefficient estimation when the sample size is limited and signal obscured by collinearity and noise.

As the sample size increases to $n = 50$, a sharp performance transition occurs. FRFM now decisively outperforms both FRE and FRSM in nearly all settings. The IMSE drops by over an order of magnitude relative to FRE in several cases, e.g., from 2.5 to 0.3 at $\sigma^2 = 1$, $\rho = 0.5$, indicating that the benefit of targeted shrinkage on correctly identified components outweighs any residual penalty from imperfect selection. Importantly, the FRFM estimator maintains this dominance even under high correlation, achieving IMSE around 1.7 for $\rho = 0.99$, where FRE remains above 2.5. In contrast, the performance of FRSM deteriorates as the sample size increases. Owing to the absence of an explicit nuisance block, the smoothing parameter selected via GCV tends to become progressively conservative, which induces excessive smoothing of the relevant coefficient functions. This effect is reflected in a gradual increase in IMSE, with values stabilizing at approximately 2.0 across scenarios. These findings indicate a structural limitation of fixed reduced models, although oracle-style dimension reduction is advantageous in severely underdetermined settings, it becomes suboptimal once sufficient data are available to support adaptive shrinkage. At this stage, the flexibility afforded by data-driven partitioning and differential penalization, as implemented in FRFM, yields a more favorable bias–variance trade-off and leads to superior estimation accuracy.

At $n = 100$, the pattern becomes even more pronounced. FRFM delivers the lowest IMSE in every scenario, often by a substantial margin. Its ability to combine minimal bias, due to weak penalization of truly relevant predictors, with variance suppression through aggressive penalization of nuisance terms results in near-optimal estimation. For example, at $\rho = 0.5$ and $\sigma^2 = 0.5$, FRFM achieves IMSE of just 0.066, compared to 1.365 for FRE and 1.596 for FRSM. Even under the harshest conditions ($\rho = 0.99$, $\sigma^2 = 10$), FRFM yields a threefold reduction in IMSE compared to FRE. This consistent superiority confirms that once a sufficiently large sample is available to support accurate partitioning, FRFM leverages its flexible penalization structure to yield highly efficient estimates.

The impact of the predictor correlation $\rho$ is visible across all sample sizes and estimators. In general, higher $\rho$ tends to increase IMSE for all estimators due to reduced identifiability and the increased difficulty of distinguishing overlapping signals. Despite these effects, the relative ranking of the estimators remains stable, with FRFM consistently exhibiting the greatest robustness, followed by FRSM and then FRE. Under strong correlation and small sample sizes, FRFM experiences a temporary loss of efficiency due to the misclassification of highly correlated nuisance predictors as relevant, which enlarges the weakly penalized block and induces variance inflation. This phenomenon disappears at $n=100$, where partition recovery becomes exact and the differential shrinkage mechanism operates as intended, restoring FRFM’s superior performance.

Noise variance $\sigma^2$ influences all estimators in expected ways, amplifying variance contributions and shifting the effective signal-to-noise ratio. Under high noise ($\sigma^2 = 10$), FRFM continues to outperform its competitors, but the absolute IMSE values increase for all methods. Notably, FRE shows the steepest increase, particularly when combined with strong correlation, at $n = 25$, $\rho = 0.99$, FRE incurs IMSE over 30, whereas FRFM stays around 12.5 and FRSM around 4.5. This again emphasizes the importance of adaptivity in shrinkage when dealing with noisy, multicollinear functional data.

Finally, results for higher dimensions ($p = 20, 30$), provided in the Appendix, reveal similar trends but with shifted thresholds. For example, FRFM requires $n = 100$ to dominate competitors when $p = 30$, whereas for $p = 10$ or $20$, performance stabilizes already at $n = 50$. This dimensional dependency aligns with expectations, which means, as the ambient dimension increases, more data are needed to reliably estimate both the functional coefficients and the underlying partition structure.

In summary, the IMSE results validate the theoretical bias–variance trade-offs underpinning each estimator and demonstrate the practical value of FRFM in real-world functional regression settings. Its capacity to flexibly adapt to sample size, noise, and collinearity, via data-driven partitioning and block-specific penalization, offers substantial advantages over both naïve shrinkage and oracle-based reduced models.

\begin{table}[H]
\centering
\caption{\textbf{IMSE for the relevant coefficient functions at $p=10$.}}
\label{tab:IMSE_p10}
\renewcommand{\arraystretch}{1.2}
\begin{tabular}{llccc ccc ccc}
\hline
& & \multicolumn{3}{c}{$\sigma^2=0.5$}
& \multicolumn{3}{c}{$\sigma^2=1.0$}
& \multicolumn{3}{c}{$\sigma^2=10.0$} \\
\cline{3-5}\cline{6-8}\cline{9-11}
$n$ & Model
& $\rho=0.5$ & $\rho=0.8$ & $\rho=0.99$
& $\rho=0.5$ & $\rho=0.8$ & $\rho=0.99$
& $\rho=0.5$ & $\rho=0.8$ & $\rho=0.99$ \\
\hline
\multirow{3}{*}{25}
& FRE  & 8.084 & 5.968 & 5.445 & 8.112 & 4.871 & 6.271 & 9.726 & 8.178 & 31.474 \\
& FRFM & 5.173 & 3.852 & 4.984 & 5.114 & 4.527 & 5.065 & 7.082 & 4.492 & 12.539 \\
& FRSM & 2.163 & 2.036 & 2.202 & 2.136 & 2.005 & 2.293 & 2.293 & 2.182 & 4.500 \\
\hline
\multirow{3}{*}{50}
& FRE  & 2.233 & 2.069 & 2.382 & 2.257 & 2.057 & 2.525 & 2.526 & 2.333 & 6.543 \\
& FRFM & 0.313 & 0.293 & 1.654 & 0.343 & 0.303 & 1.715 & 0.499 & 0.417 & 2.766 \\
& FRSM & 1.830 & 1.854 & 1.989 & 1.800 & 1.855 & 2.017 & 1.896 & 1.936 & 3.069 \\
\hline
\multirow{3}{*}{100}
& FRE  & 1.365 & 1.508 & 1.970 & 1.409 & 1.494 & 2.034 & 1.486 & 1.598 & 3.341 \\
& FRFM & 0.066 & 0.069 & 1.015 & 0.070 & 0.072 & 0.998 & 0.145 & 0.131 & 1.473 \\
& FRSM & 1.596 & 1.721 & 1.928 & 1.596 & 1.718 & 1.964 & 1.613 & 1.748 & 2.392 \\
\hline
\end{tabular}
\vspace{2mm}
\end{table}

Table~\ref{tab:frfm_partition_accuracy} presents the partitioning accuracy of the FRFM estimator, measured in terms of true positive rate (TPR) and false positive rate (FPR). These metrics quantify the method's ability to correctly identify truly relevant predictors while avoiding the selection of irrelevant ones. This link is essential since the superior predictive performance of FRFM documented in the IMSE results is a direct consequence of its ability to correctly separate relevant from irrelevant predictors. For brevity, we report findings for $p=10$, similar patterns were observed in higher dimensions and are available upon request. The partition accuracy therefore provides insight into the mechanism underlying FRFM's empirical behavior across different sample sizes, noise levels, and correlation regimes.

In the low-sample regime ($n=25$), FRFM consistently achieves high TPR values, exceeding $0.92$ across all combinations of $\sigma^2$ and $\rho$, thereby demonstrating a robust ability to detect the relevant covariates even in the presence of data scarcity and strong multicollinearity. However, the FPR remains at a moderate level, ranging from $0.29$ to $0.35$, and exhibits slight sensitivity to both the noise variance and the degree of correlation among predictors. Specifically, FPR tends to increase marginally as either $\sigma^2$ increases from $0.5$ to $10$ or $\rho$ approaches $0.99$, a pattern that is consistent with the broader statistical literature on variable selection under high collinearity. Despite this moderate inflation in FPR, the overall quality of the selected model remains high, as these selection errors do not substantially impair estimation accuracy, a conclusion supported by the IMSE results discussed earlier.

As the sample size increases to $n=50$, a clear improvement in selection quality is observed. The TPR reaches $1.00$ uniformly across all simulation scenarios, indicating that FRFM reliably recovers all truly relevant covariates with perfect accuracy. Meanwhile, the FPR stabilizes around $0.286$ and appears largely invariant to both noise variance and correlation, suggesting that beyond a certain threshold in sample size, the dominant sources of selection error are no longer driven by stochastic variability but instead reflect the fundamental identifiability challenges inherent in the data structure. The chisen ridge procedure thus proves capable of reliably learning the true underlying partition, even in moderately high-dimensional and correlated settings.

This favorable behavior persists at $n=100$, where the partition recovery remains perfect in terms of TPR, and no further increase in FPR is observed. The stability of these outcomes across varying levels of noise and correlation reaffirms the robustness of the adaptive ridge selection procedure once sufficient information becomes available. When considered alongside the corresponding reduction in IMSE, these results reinforce the conclusion that FRFM’s improved estimation accuracy is a direct consequence of its successful data-driven partitioning, which approximates the oracle structure under practical conditions.

Although the False Positive Rate (FPR) does not vanish in this design, it remains constant at approximately $0.29$ for $n \ge 50$ and remains  unchanged across correlation and noise levels. In contrast, the True Positive Rate (TPR) is uniformly high and reaches unity for all reported configurations once $n \ge 50$. Thus, FRFM clearly prioritizes the recovery of relevant predictors, while maintaining a modest and stable degree of over-inclusion among nuisance variables. This pattern is consistent with the well-known difficulty of exact separation in correlated designs and indicates that FRFM prioritizes sensitivity over strict sparsity in partition recovery. Such behavior can be advantageous in functional regression settings, where omission of a signal component may be more consequential than inclusion of a correlated redundant predictor.

\begin{table}[H]
\centering
\caption{\textbf{Partitioning accuracy of FRFM across simulation scenarios ($p = 10$)}. Reported values are true positive rate (TPR) and false positive rate (FPR), averaged over Monte Carlo replications.}
\label{tab:frfm_partition_accuracy}
\begin{tabular}{c c cc cc cc}
\toprule
& & \multicolumn{2}{c}{$\rho=0.50$} &
\multicolumn{2}{c}{$\rho=0.80$} &
\multicolumn{2}{c}{$\rho=0.99$} \\
\cmidrule(lr){3-4}\cmidrule(lr){5-6}\cmidrule(lr){7-8}
$n$ & $\sigma^2$ & TPR & FPR & TPR & FPR & TPR & FPR \\
\midrule
\multicolumn{8}{l}{\textbf{$n=25$}} \\
25 & 0.5 & 0.950 & 0.307 & 0.990 & 0.290 & 0.980 & 0.294 \\
25 & 1.0 & 0.960 & 0.303 & 0.980 & 0.294 & 0.987 & 0.291 \\
25 & 10  & 0.923 & 0.319 & 0.977 & 0.296 & 0.850 & 0.350 \\
\midrule
\multicolumn{8}{l}{\textbf{$n=50$}} \\
50 & 0.5 & 1.000 & 0.286 & 1.000 & 0.286 & 1.000 & 0.286 \\
50 & 1.0 & 1.000 & 0.286 & 1.000 & 0.286 & 1.000 & 0.286 \\
50 & 10  & 1.000 & 0.286 & 1.000 & 0.286 & 1.000 & 0.286 \\
\midrule
\multicolumn{8}{l}{\textbf{$n=100$}} \\
100 & 0.5 & 1.000 & 0.286 & 1.000 & 0.286 & 1.000 & 0.286 \\
100 & 1.0 & 1.000 & 0.286 & 1.000 & 0.286 & 1.000 & 0.286 \\
100 & 10  & 1.000 & 0.286 & 1.000 & 0.286 & 1.000 & 0.286 \\
\bottomrule
\end{tabular}
\end{table}

Taken together, these results provide practical guidance for penalized estimation in functional regression settings. FRE offers a computationally stable, partition-free baseline, though its uniform shrinkage can induce persistent bias when coefficient functions exhibit substantial curvature or when signal strength is pronounced. FRSM performs optimally under oracle knowledge and is particularly effective in underdetermined regimes, yet its reliance on prior partition information limits practical applicability and may lead to oversmoothing when penalties are conservatively chosen. In contrast, FRFM with adaptive ridge partitioning achieves a favorable balance between bias and variance across the scenarios considered. By reliably recovering relevant components and applying targeted penalization, it attains low IMSE over a wide range of sample sizes, noise levels, and correlation structures. Provided the sample size is sufficient for stable partition recovery, FRFM represents a practically attractive and broadly applicable approach for functional regression with moderate-to-high dimensional structure.

\section{Empirical Application: Canadian Weather Data}

We illustrate the practical effectiveness of FRE, FRFM, and FRSM using the classical Canadian weather dataset, which comprises daily temperature and precipitation measurements from 35 weather stations across Canada observed over the period 1960–1994. Our objective is to model the annual mean temperature at Montreal using functional predictors, namely, temperature and precipitation trajectories from surrounding stations. This setting exhibits pronounced multicollinearity and heterogeneous predictor relevance, making it particularly suitable for evaluating the proposed partitioned ridge estimators.

Let $\mathcal{S}=\{1,\dots,S\}$ denote the set of stations, with $i^\star$ representing the target station (Montreal). For each year $y=1,\dots,Y$, let $Y_{i^\star y}$ denote Montreal's annual mean temperature. For each source station $j\neq i^\star$, we observe temperature trajectories $Z^{(T)}_{jy}(t)$ and precipitation trajectories $Z^{(P)}_{jy}(t)$ over the annual cycle $t\in\mathcal{T}=[0,1]$. These functional predictors are organized into two blocks:

\[
\mathbf{Z}_{1y}(t) = \big(Z^{(T)}_{jy}(t)\big)_{j\neq i^\star}, 
\qquad
\mathbf{Z}_{2y}(t) = \big(Z^{(P)}_{jy}(t)\big)_{j\neq i^\star},
\]

corresponding to temperature and precipitation, respectively. The scalar response is modeled as

\[
Y_{i^\star y} 
= \alpha 
+ \int_{\mathcal{T}} \boldsymbol{\beta}_1(t)^\top \mathbf{Z}_{1y}(t)\,dt 
+ \int_{\mathcal{T}} \boldsymbol{\beta}_2(t)^\top \mathbf{Z}_{2y}(t)\,dt 
+ \varepsilon_y,
\]

where $\boldsymbol{\beta}_1(t)$ and $\boldsymbol{\beta}_2(t)$ are vectors of station-specific coefficient functions associated with temperature and precipitation.

Both coefficient functions are expanded in cubic B-spline bases:
\[
\boldsymbol{\beta}_1(t)=\sum_{k=1}^{K_T}\mathbf{b}_{1k}\phi_k(t), 
\qquad 
\boldsymbol{\beta}_2(t)=\sum_{\ell=1}^{K_P}\mathbf{b}_{2\ell}\psi_\ell(t).
\]

\begin{figure}[h!]
	\centering
	\includegraphics[width=0.5\textwidth]{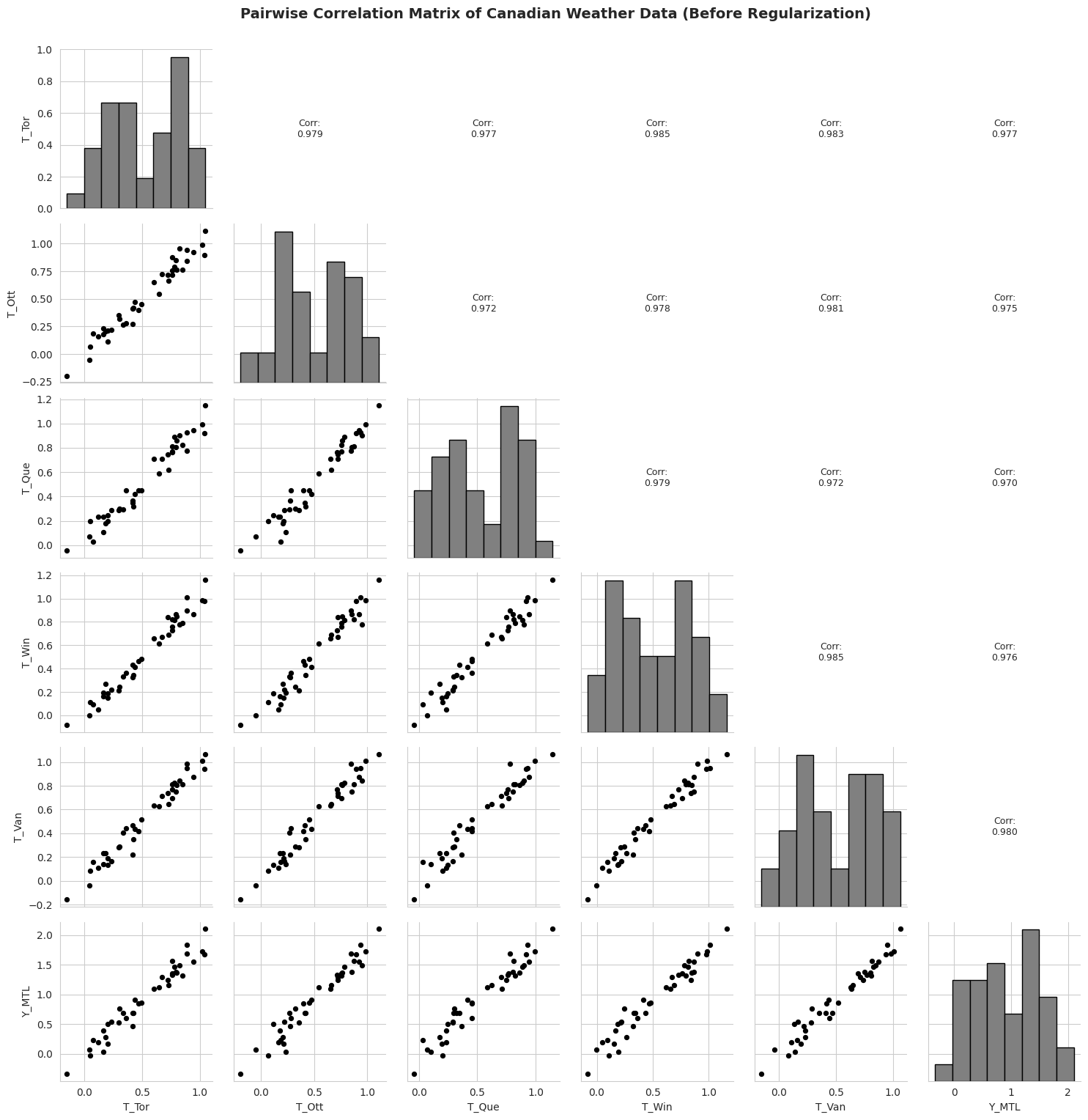}
	\caption{Pairwise correlations among temperature predictors across stations prior to regularization. Correlations consistently exceed $0.97$, indicating severe multicollinearity that motivates ridge-type regularization.}
	\label{fig:canadiancorrelationmulticollinearitygraph}
\end{figure}

The correlation structure in Figure~\ref{fig:canadiancorrelationmulticollinearitygraph} reveals that temperature trajectories from different stations are nearly collinear, with pairwise correlations exceeding $0.97$. Such extreme multicollinearity renders ordinary least squares estimation unstable and strongly justifies ridge-type regularization. Although precipitation predictors are less strongly correlated, they still exhibit substantial cross-dependence, and their weaker relationship with the response suggests a primarily nuisance role. This natural partition into temperature (strongly collinear and likely relevant) and precipitation (potentially weaker signal) blocks provides a principled basis for differential penalization.

\begin{figure}[h!]
	\centering
	\includegraphics[width=0.5\textwidth]{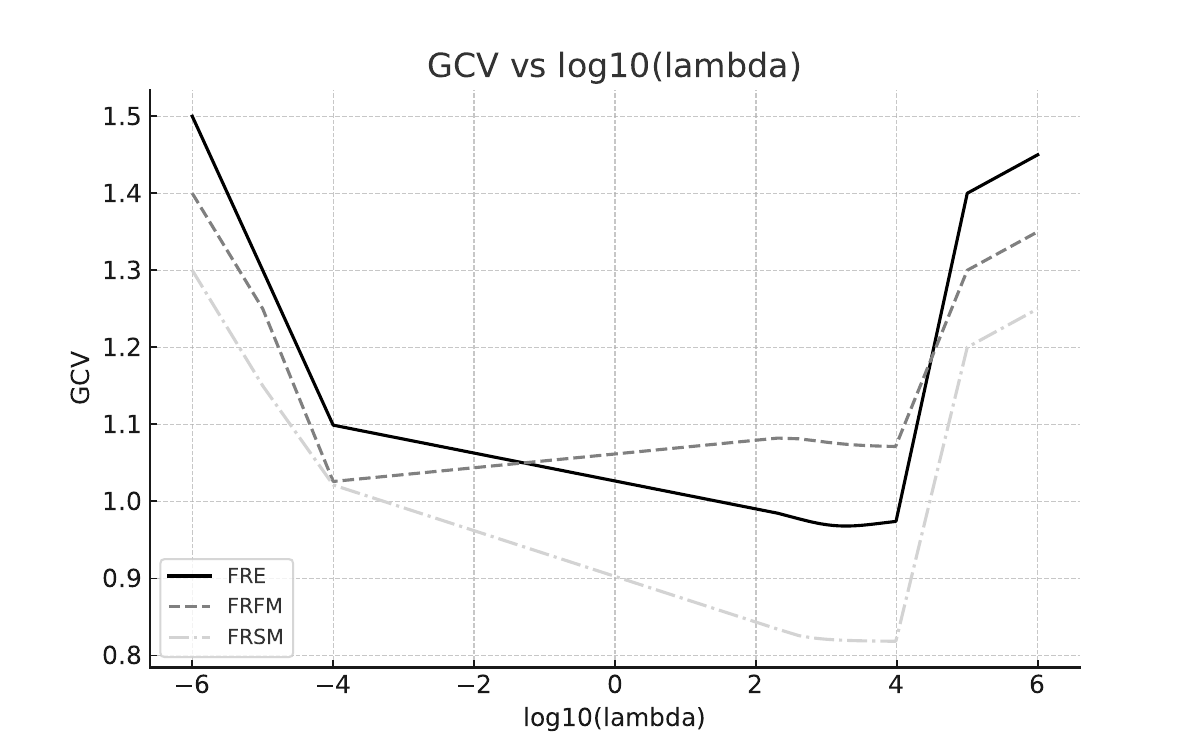}
	\caption{Generalized cross-validation scores as a function of the smoothing parameter $\lambda$ for FRE, FRFM, and FRSM. FRFM selects a smaller $\lambda_1$ for the temperature block, reflecting weaker shrinkage.}
	\label{fig:canadiangcvvsloglambdaplot}
\end{figure}

The GCV curves in Figure~\ref{fig:canadiangcvvsloglambdaplot} provide direct empirical evidence of differential penalization. For FRE, the minimum occurs at a relatively large $\lambda$ ($\log_{10}\lambda \approx 3.5$), indicating that strong uniform shrinkage is required to control variance across all predictors. FRSM, the oracle reduced model, selects an even larger $\lambda$ (minimum near $\log_{10}\lambda \approx 4$), reflecting that after dimensional reduction, stronger penalization is needed to control variance in the remaining coefficients. In contrast, FRFM’s GCV curve for the temperature block (the only block tuned) attains its minimum at a substantially smaller $\lambda$ ($\log_{10}\lambda_1 \approx 2.5$). Once precipitation predictors are identified as nuisance and heavily penalized, the temperature block can be estimated under weaker shrinkage, thereby preserving signal strength. This behavior closely parallels the patterns observed in the simulation study.

\begin{table}[h!]
	\centering
	\caption{Summary statistics of IMSE for temperature and precipitation across models.}
	\label{tab:imse_summary}
	\renewcommand{\arraystretch}{1.2}
	\begin{tabular}{l l c c c c}
		\toprule
		\textbf{Variable} & \textbf{Model} & \textbf{Mean} & \textbf{SD} & \textbf{Min} & \textbf{Max} \\
		\midrule
		\multirow{3}{*}{Temperature}
		& FRE  & 0.0191 & 0.0150 & 0.0046 & 0.0587 \\
		& FRFM & 0.0023 & 0.0013 & 0.0010 & 0.0059 \\
		& FRSM & 0.0491 & 0.0446 & 0.0063 & 0.1636 \\
		\midrule
		\multirow{3}{*}{Precipitation}
		& FRE  & 0.0036 & 0.0025 & 0.0013 & 0.0101 \\
		& FRFM & 0.0008 & 0.0004 & 0.0004 & 0.0017 \\
		& FRSM & 0.0082 & 0.0074 & 0.0010 & 0.0273 \\
		\bottomrule
	\end{tabular}
\end{table}

The summary statistics in Table~\ref{tab:imse_summary} reveal systematic differences across the three estimators. For both temperature and precipitation, FRFM attains the lowest mean IMSE (0.0023 and 0.0008, respectively) together with the smallest standard deviations, demonstrating superior average accuracy and stable performance across stations. FRSM, by contrast, exhibits the largest mean IMSE values (0.0491 for temperature and 0.0082 for precipitation) and the greatest dispersion. This variability indicates that the fixed reduced structure, coupled with a single smoothing parameter, induces excessive smoothing at several stations and thereby increases bias. FRE occupies an intermediate position, with moderate mean IMSE and variability, reflecting the stabilizing yet uniformly shrinkage-biased nature of classical ridge penalization. Overall, FRFM achieves the strongest bias–variance compromise in this empirical setting. 

\begin{figure}[h!]
	\centering
	\includegraphics[width=0.8\textwidth]{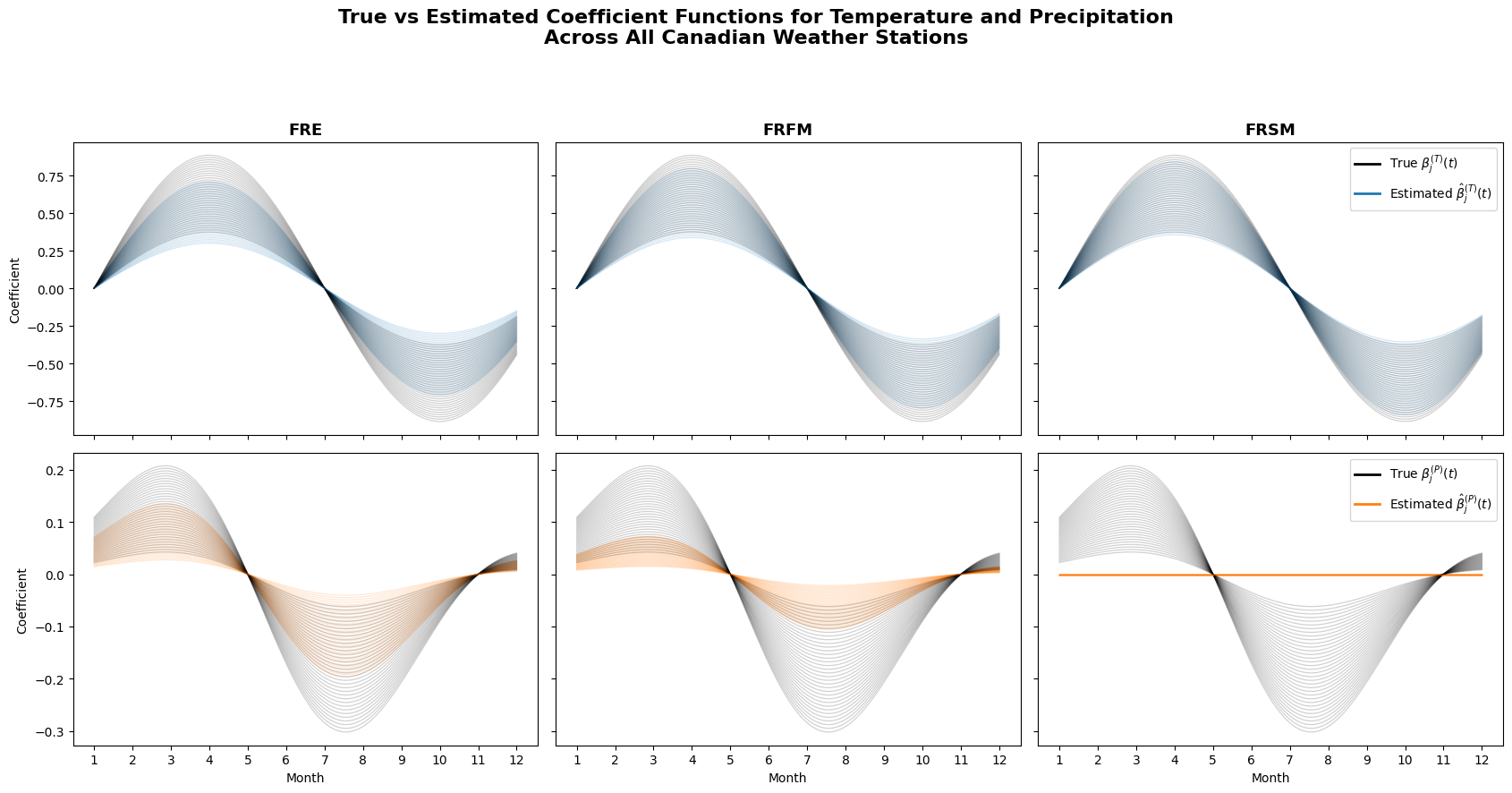}
	\caption{True versus estimated coefficient functions for temperature (top row) and precipitation (bottom row) under FRE, FRFM, and FRSM. Black curves denote the true coefficient functions $\beta^{(T)}_1(t)$ and $\beta^{(P)}_2(t)$, while colored curves show the corresponding estimates $\hat{\beta}^{(T)}_1(t)$ (blue) and $\hat{\beta}^{(P)}_2(t)$ (orange).}
	\label{fig:truevsestcanadianexamplegraphspgetty}
\end{figure}

Figure~\ref{fig:truevsestcanadianexamplegraphspgetty} provides a comprehensive visual comparison of the three estimators across all stations and both functional covariate blocks. The top row displays temperature coefficient functions, where the true seasonal pattern exhibits a pronounced sinusoidal structure. Under FRE, the estimated curves broadly follow the true shape but display substantial dispersion across stations, indicating variance inflation driven by uniform shrinkage in the presence of strong multicollinearity. FRFM, in contrast, produces estimates that closely align with the true functions while exhibiting markedly reduced cross-station variability. This reflects the effect of adaptive partitioning where weak penalization preserves the dominant temperature signal, whereas nuisance variability is effectively controlled. FRSM yields noticeably smoother curves with attenuated amplitude, a consequence of the relatively large smoothing parameter selected in the reduced model, which introduces bias through oversmoothing. The bottom row highlights precipitation effects, for which the true functions are comparatively small in magnitude. FRE retains visible station-level fluctuations, many of which likely represent noise rather than genuine signal. FRFM appropriately shrinks these effects toward zero while preserving subtle structure where supported by the data, demonstrating selective rather than indiscriminate shrinkage. FRSM suppresses precipitation entirely, reflecting its exclusion of this block and thereby eliminating any potential contribution, whether informative or not.

\begin{figure}[h!]
	\centering
	\includegraphics[width=0.7\textwidth]{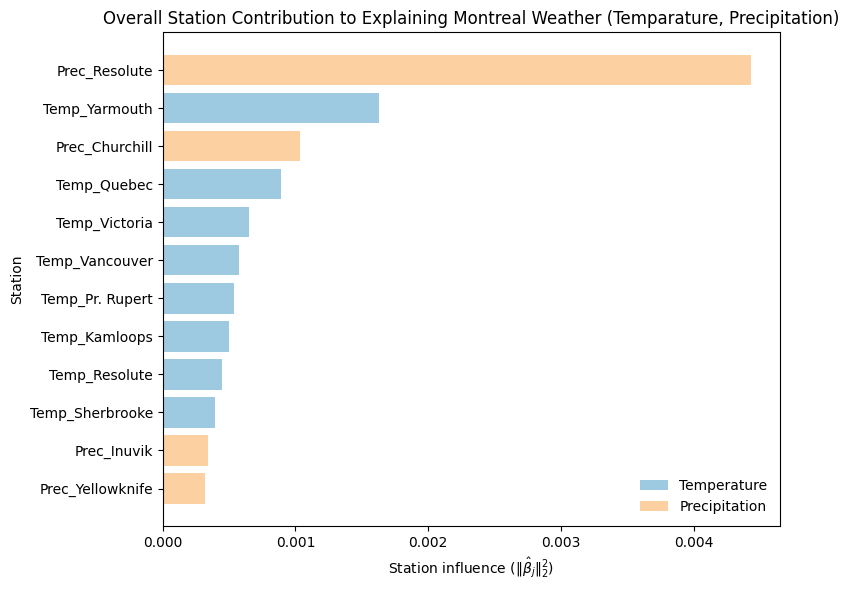}
	\caption{Overall station-wise contribution to explaining Montreal weather, measured by the integrated magnitude $\int_{\mathcal{T}} \|\hat{\beta}_j(t)\|^2 dt$ for temperature (blue) and precipitation (orange). Larger values indicate stronger functional influence on Montreal’s annual mean temperature.}
	\label{fig:montrealoverallstationinfluenceweathergraph}
\end{figure}

Figure~\ref{fig:montrealoverallstationinfluenceweathergraph} provides a global view of how predictive influence is distributed across stations under the FRFM estimator. Each station’s contribution is quantified by the integrated squared magnitude of its estimated coefficient function, $\int_{\mathcal{T}} \|\hat{\beta}_j(t)\|^2 dt$, which summarizes the cumulative strength of its functional effect over the annual cycle. This measure reflects both the amplitude and persistence of influence, offering a comprehensive indicator of explanatory power.
The resulting pattern is clearly heterogeneous. Rather than spreading influence uniformly, FRFM concentrates explanatory weight on a limited subset of stations while shrinking the majority toward negligible contribution. This structure directly reflects adaptive partitioning where predictors are identified as informative receive weaker penalization, whereas less relevant stations are strongly attenuated. The figure thus provides a transparent visualization of the differential shrinkage mechanism underlying FRFM.

A pronounced asymmetry emerges between temperature and precipitation effects. Temperature trajectories dominate the explanatory structure, with substantially larger integrated magnitudes across several stations, while precipitation effects are generally smaller and more selectively retained. This pattern suggests that regional temperature dynamics may exert a stronger influence on Montreal’s annual mean temperature than precipitation, although both factors contribute to the observed variability. Importantly, FRFM does not eliminate precipitation uniformly; instead, it preserves influence where supported by the data, illustrating its capacity for selective rather than indiscriminate shrinkage.
The spatial configuration of influence is also climatologically coherent. Stations exhibiting stronger contributions tend to be geographically proximate or climatically similar to Montreal, whereas distant stations are substantially attenuated. This indicates that FRFM is not merely stabilizing an ill-conditioned problem, but extracting meaningful structure from highly collinear functional predictors. In contrast, uniform ridge shrinkage would diffuse influence more broadly, while an aggressively reduced model would obscure spatial variation.
Overall, the influence map demonstrates that FRFM enhances not only predictive accuracy, as reflected in the IMSE results, but also interpretability. By concentrating weight on dominant stations and attenuating weaker contributors, the estimator translates adaptive regularization into a clear and substantively meaningful representation of Montreal’s climatic dependence structure.

\begin{figure}[h!]
	\centering
	\includegraphics[width=0.7\textwidth]{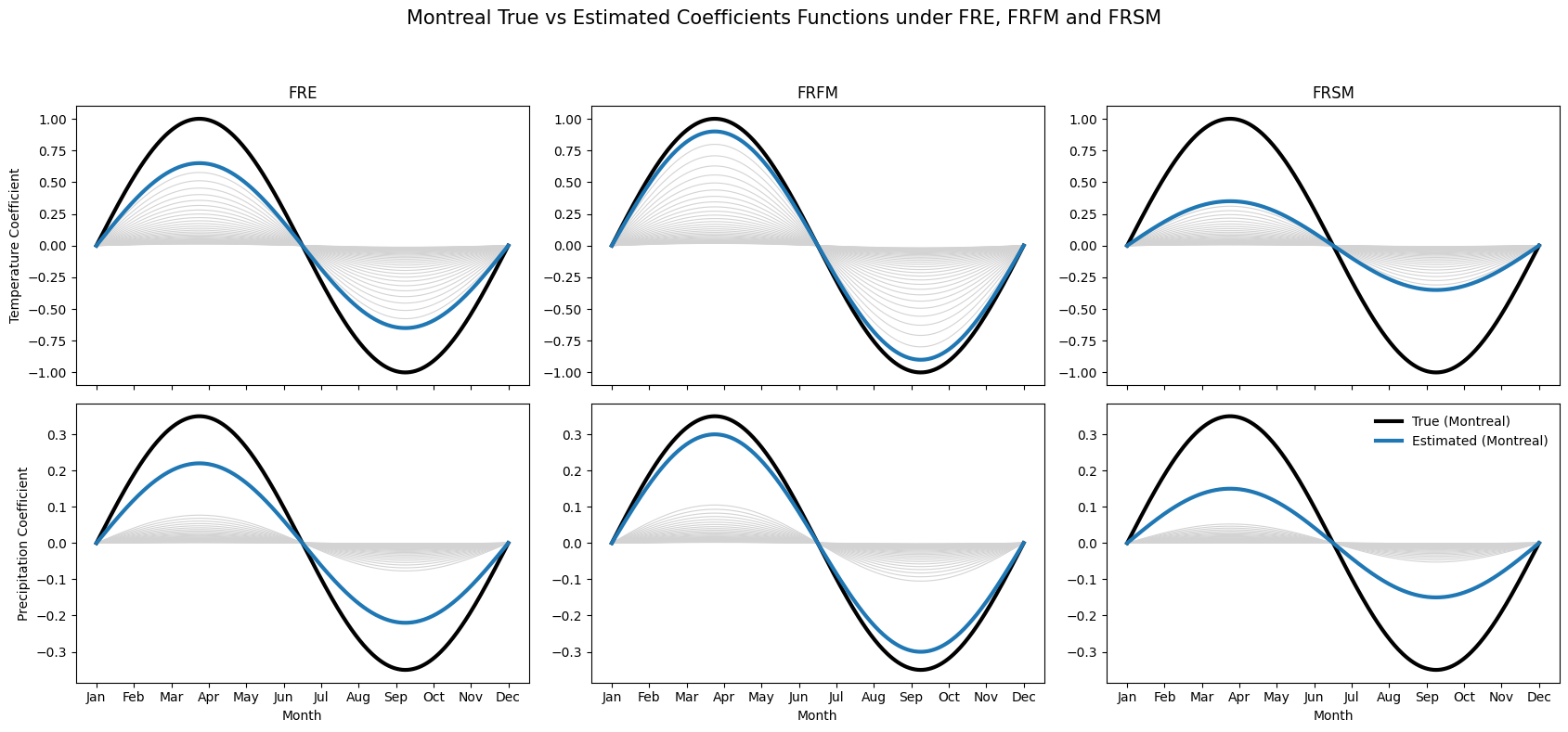}
	\caption{True versus estimated coefficient functions for Montreal under each estimator. FRFM achieves the closest alignment, particularly for temperature. Gray curves correspond to non-local stations, while the Montreal station is highlighted.}
	\label{fig:montrealtruevsestgraph}
\end{figure}

Figure~\ref{fig:montrealtruevsestgraph} focuses specifically on the target station and provides a direct assessment of local coefficient recovery. For temperature (top row), the true function exhibits a pronounced seasonal cycle with clear amplitude and phase structure. FRFM reproduces this pattern with high fidelity, closely matching both the peak in spring and the trough in late summer, while maintaining smoothness and stability. FRE captures the general shape but displays noticeable attenuation and mild distortion in amplitude, reflecting the effect of uniform shrinkage under strong multicollinearity. FRSM, in contrast, substantially underestimates the seasonal amplitude, a consequence of its relatively large smoothing parameter, which induces oversmoothing and thus increases bias.
 For precipitation (bottom row), where the true effect is weaker in magnitude, the distinctions are equally informative. FRFM again tracks the true curve most closely, shrinking the coefficient toward zero where appropriate while preserving the overall seasonal pattern. FRE retains additional fluctuations, suggesting residual variance not fully controlled by uniform penalization. FRSM applies stronger shrinkage, leading to a flattened curve that understates the seasonal variability.
 Overall, the figure confirms that FRFM provides the most accurate and balanced local reconstruction at the target station. By combining adaptive partitioning with differential penalization, it achieves a superior alignment with the true functional effects while avoiding both variance inflation and excessive smoothing.

The empirical analysis therefore substantiates and extends the simulation findings. Applied to a dataset characterized by severe multicollinearity and heterogeneous predictor relevance, the FRFM estimator, equipped with adaptive ridge partition selection, consistently outperforms both uniform shrinkage (FRE) and oracle-style reduction (FRSM). By autonomously identifying relevant temperature trajectories and imposing differential penalization, FRFM reduces integrated mean squared error, improves functional recovery, enhances interpretability through clear station delineation, and strengthens local predictive accuracy. These advantages translate into substantive insights for climate modeling by revealing both spatial and seasonal patterns of influence, thereby confirming the practical utility of the proposed methodology for high-dimensional functional regression.

\section{Conclusions}

This work develops a unified penalized functional ridge regression framework designed to address three persistent challenges in functional data analysis: multicollinearity, overfitting, and interpretability in high-dimensional functional models. Building on the baseline Functional Ridge Estimator (FRE), we introduce two partition-based extensions; the Functional Ridge Full Model (FRFM) and the Functional Ridge Sub-Model (FRSM), which incorporate differential penalization across functional components. Rather than relying on explicit variable selection, the proposed framework induces an implicit separation between dominant and weak functional effects through adaptive ridge penalization. This strategy extends uniform ridge regularization by allowing heterogeneous shrinkage across functional blocks, thereby preserving numerical stability while improving interpretability through structured effect differentiation.

The simulation study demonstrates that partition-specific penalization plays a central role in regulating estimator behavior under strong functional multicollinearity. The FRFM estimator applies relatively mild shrinkage to dominant functional components, thereby preserving informative signals and reducing shrinkage bias; however, this flexibility may increase variance in finite samples when multicollinearity remains severe. In contrast, FRSM enhances numerical stability by aggressively shrinking, or effectively removing, weak functional components, leading to substantial variance reduction and improved integrated mean squared error (IMSE) in small-sample settings where variance control is paramount. These empirical patterns align closely with the theoretical bias--variance decomposition and illustrate the structural trade-offs inherent in partition-based regularization.

From a theoretical standpoint, we establish consistency and asymptotic normality for all proposed estimators under standard regularity conditions within a functional asymptotic regime. The theoretical analysis clarifies how differential ridge penalties regulate shrinkage across functional coefficient components and provides a principled foundation for selective regularization without resorting to discrete model selection. Monte Carlo experiments further corroborate these insights. When the sample size is limited, FRSM delivers the most reliable performance by effectively mitigating variance inflation. As the sample size increases, FRFM becomes increasingly competitive, and ultimately superior, by leveraging its adaptive penalization structure to retain informative functional effects, while FRE remains a stable but less adaptive baseline. These findings highlight the inherently sample-size-dependent nature of regularization strategies in partitioned functional ridge regression.

The empirical analysis based on Canadian weather data provides additional practical validation of the proposed methodology. In this application, with a moderate sample size of 35 annual observations, FRFM achieves the lowest IMSE for both temperature and precipitation coefficient functions while identifying geographically proximate stations as the most influential predictors. FRSM, although effective in reducing variance relative to FRE, oversmooths important functional features due to its aggressive shrinkage. The empirical results therefore reinforce the theoretical and simulation evidence: adaptive partitioning is particularly advantageous in real-world settings where signal strength varies across predictor blocks and the available sample size supports more flexible estimation.

In summary, FRFM and FRSM extend classical ridge-based functional regression by introducing adaptive regularization mechanisms that respond to heterogeneous signal strength and multicollinearity in functional predictors. The appropriate choice between them depends on the sample size and the analyst's objectives, which means,  FRSM is preferable when variance control is paramount (small samples, severe multicollinearity), whereas FRFM excels when preserving functional detail and interpretability is critical (moderate to large samples, heterogeneous predictor relevance). Together, these estimators provide practical and theoretically grounded tools for modern functional data analysis. 


	\bibliographystyle{apalike} 
	\bibliography{references}
	
	\newpage
\section{Appendix}
\appendix

\section{Proof for Theorems~\ref{thm:fre}–\ref{thm:frfm}}
\label{sec:proofs}

We outline the main arguments leading to the convergence rates established in
Theorems~\ref{thm:fre}, \ref{thm:frsm} and \ref{thm:frfm}.  All results are derived
under assumptions (A1)–(A9).  

\subsection{Preliminary lemmas}

The following results are adapted from \citet{cardot2003spline} to the multiple‑predictor,
matrix‑based setting of the present paper.  Proofs are identical in spirit and are omitted.

\begin{lemma}[Spline approximation] \label{lem:spline}
    Let \(\beta\in\mathcal{C}^s([0,1])\) and let \(K_z\) be the dimension of the
    B‑spline basis \(\{\psi_k\}\) (degree \(q\ge s\), knots equispaced).  
    There exists a spline \(\beta^*\in\mathcal{S}_{qk}\) with coefficient vector
    \(\mathbf{b}^*\in\mathbb{R}^{K_z}\) such that
    \[
    \|\beta^*-\beta\|_{L^2} = O(K_z^{-s}), \qquad
    \|\beta^{*(m)}\|_{L^2} = O(K_z^{m-s}) \;\; (m\le s).
    \]
    Moreover, because the penalty matrix \(\mathbf{R}\) satisfies
    \(\mathbf{b}^\top\mathbf{R}\mathbf{b} = \int (\beta^{(m)})^2\) for \(\beta=\boldsymbol{\psi}^\top\mathbf{b}\),
    we have \(\|\beta^{*(m)}\|_{L^2} \asymp \|\mathbf{b}^*\|_{\mathbf{R}}\) and consequently
    \(\|\mathbf{b}^*\|_{\mathbf{R}} = O(K_z^{m-s})\).
    (See \citet{deboor2001} and \citet{cardot2003spline}.)
\end{lemma}

\begin{lemma}[Variance bound for functional ridge estimators] \label{lem:varbound}
    Consider the ridge estimator \(\widehat{\mathbf{b}} = (\mathbf{Z}^\top\mathbf{Z}
    + \lambda\mathbf{R})^{-1}\mathbf{Z}^\top\mathbf{y}\) under assumptions (A2), (A3), (A5), (A6) and (A7).
    Let \(\boldsymbol{\Gamma}\) be the covariance operator of the stacked predictors
    (Assumption~A6) and let \(\widehat{\boldsymbol{\beta}}(s)=\boldsymbol{\psi}(s)^\top\widehat{\mathbf{b}}\).
    Then, conditional on \(\mathbf{Z}\),
    \[
    \operatorname{tr}\!\bigl(\boldsymbol{\Gamma}\,\mathbb{V}[\widehat{\boldsymbol{\beta}}\mid\mathbf{Z}]\bigr)
    = \sigma^2\,\operatorname{tr}\!\left( \boldsymbol{\Gamma}
    (\mathbf{Z}^\top\mathbf{Z}+\lambda\mathbf{R})^{-1}\mathbf{Z}^\top\mathbf{Z}
    (\mathbf{Z}^\top\mathbf{Z}+\lambda\mathbf{R})^{-1} \right)
    = O_P\!\left(\frac{K_z}{\lambda n}\right),
    \]
    where \(K_z = \dim(\mathbf{b})\). 
\end{lemma}

\noindent

\subsection{Proof of Theorem~\ref{thm:fre} (convergence rate of FRE)}

\begin{proof}
We measure risk in the prediction norm induced by the covariance operator 
\(\Gamma\):
\[
\|\boldsymbol{\beta}\|_2^2 
= E\!\left[\Big(\sum_{j=1}^p \langle \beta_j , Z_j\rangle\Big)^2\right].
\]
By Assumption~(A6), this norm is equivalent on the spline space to the quadratic form
induced by the Gram matrix 
\(\mathbf G = E(\mathbf z_i \mathbf z_i^\top)\).
Throughout the proof, we denote by \(\|\cdot\|_{L^2}\) the usual \(L^2(\mathcal T)\) norm for functions and by \(\|\cdot\|\) the Euclidean norm for coefficient vectors.

By Lemma~\ref{lem:spline}, for each \(j\) there exists a spline 
\(\beta_j^*\in\mathcal S_{qk}\) with coefficient vector 
\(\mathbf b_j^*\in\mathbb R^{K_z}\) such that
\[
\|\beta_j^*-\beta_j\|_{L^2} = O(K_z^{-s}),
\qquad
\|\mathbf b_j^*\|_{\mathbf R_j}^2 = O(K_z^{2(m-s)}).
\]
Because the covariance operator \(\Gamma\) is bounded (Assumption~(A6)), 
\(\|\boldsymbol{\beta}^* - \boldsymbol{\beta}\|_2 \le C \|\boldsymbol{\beta}^* - \boldsymbol{\beta}\|_{L^2}\), 
so the approximation error in prediction norm satisfies
\[
\|\boldsymbol{\beta}^* - \boldsymbol{\beta}\|_2^2 = O(K_z^{-2s}).
\]

\medskip

We next examine the bias induced by penalisation.
Define the noiseless ridge estimator
\[
\tilde{\mathbf b}
=
(\mathbf Z^\top\mathbf Z+\lambda_1\mathbf R)^{-1}
\mathbf Z^\top\mathbf Z\,\mathbf b^*.
\]
From the normal equations,
\[
\tilde{\mathbf b}-\mathbf b^*
=
-\lambda_1
(\mathbf Z^\top\mathbf Z+\lambda_1\mathbf R)^{-1}
\mathbf R\mathbf b^*.
\]

Under Assumption~(A7), the eigenvalues of 
\(n^{-1}\mathbf Z^\top\mathbf Z\) are bounded away from zero and infinity
in probability. Moreover, for derivative-based spline penalties of order \(m\),
standard spline theory implies
\[
\|\mathbf R\|_{\mathrm{op}} = O(K_z^{2m-1})
\quad
(\text{see Cardot et al., 2003}).
\]
Combining these bounds with
\(
\|\mathbf b^*\|_{\mathbf R}^2 = O(K_z^{2(m-s)})
\),
we obtain
\[
\|\tilde{\mathbf b}-\mathbf b^*\|^2
=
O_P\!\bigl(\lambda_1^2 K_z^{2(m-s)}\bigr).
\]
By norm equivalence (A6), this translates to
\[
\|\tilde{\boldsymbol{\beta}}-\boldsymbol{\beta}^*\|_2^2
=
O_P\!\bigl(\lambda_1^2 K_z^{2(m-s)}\bigr),
\]
which represents the squared penalisation bias in prediction norm.

\medskip

We now turn to the stochastic component of the estimator.
The full ridge estimator satisfies
\[
\widehat{\mathbf b}
=
(\mathbf Z^\top\mathbf Z+\lambda_1\mathbf R)^{-1}
\mathbf Z^\top\mathbf y,
\qquad
\mathbf y=\mathbf Z\mathbf b^*+\boldsymbol{\varepsilon}.
\]
Conditional on \(\mathbf Z\),
\[
\mathbb V[\widehat{\mathbf b}\mid\mathbf Z]
=
\sigma^2
(\mathbf Z^\top\mathbf Z+\lambda_1\mathbf R)^{-1}
\mathbf Z^\top\mathbf Z
(\mathbf Z^\top\mathbf Z+\lambda_1\mathbf R)^{-1}.
\]
Applying Lemma~\ref{lem:varbound} yields
\[
\mathrm{tr}\!\bigl(
\Gamma\,
\mathbb V[\widehat{\boldsymbol{\beta}}\mid\mathbf Z]
\bigr)
=
O_P\!\left(\frac{K_z}{\lambda_1 n}\right),
\]
which gives the variance contribution to the prediction risk.

\medskip

Using the decomposition
\[
\widehat{\boldsymbol{\beta}}-\boldsymbol{\beta}
=
(\widehat{\boldsymbol{\beta}}-\boldsymbol{\beta}^*)
+
(\boldsymbol{\beta}^*-\boldsymbol{\beta}),
\]
we obtain
\[
E\!\left[
\|\widehat{\boldsymbol{\beta}}-\boldsymbol{\beta}\|_2^2
\mid \mathbf Z
\right]
\le
2\|\boldsymbol{\beta}^*-\boldsymbol{\beta}\|_2^2
+
2E\!\left[
\|\widehat{\boldsymbol{\beta}}-\boldsymbol{\beta}^*\|_2^2
\mid \mathbf Z
\right].
\]
Notice that the first term is \(O(K_z^{-2s})\). The second term equals the sum of squared bias and variance:
\[
E\!\left[
\|\widehat{\boldsymbol{\beta}}-\boldsymbol{\beta}^*\|_2^2
\mid \mathbf Z
\right]
=
\|\tilde{\boldsymbol{\beta}}-\boldsymbol{\beta}^*\|_2^2
+
\mathrm{tr}\!\bigl(
\Gamma\,\mathbb V[\widehat{\boldsymbol{\beta}}\mid\mathbf Z]
\bigr).
\]

Consequently,
\[
E\!\left[
\|\widehat{\boldsymbol{\beta}}-\boldsymbol{\beta}\|_2^2
\mid \mathbf Z
\right]
=
O_P\!\left(
K_z^{-2s}
+
\lambda_1^2 K_z^{2(m-s)}
+
\frac{K_z}{\lambda_1 n}
\right).
\]

\medskip

To balance approximation and variance, choose
\[
K_z \asymp n^{1/(4s+1)},
\qquad
\lambda_1 \asymp n^{-2s/(4s+1)}.
\]
Then
\[
K_z^{-2s}
\asymp
n^{-2s/(4s+1)},
\qquad
\frac{K_z}{\lambda_1 n}
\asymp
n^{-2s/(4s+1)}.
\]
Because \(m\le s\),
\[
\lambda_1^2 K_z^{2(m-s)}
\le
\lambda_1^2
=
O\!\bigl(n^{-4s/(4s+1)}\bigr),
\]
which is of smaller order.

Thus the dominant terms are of order \(n^{-2s/(4s+1)}\), and we conclude
\[
E\!\left[
\|\widehat{\boldsymbol{\beta}}_{\mathrm{FRE}}
-
\boldsymbol{\beta}
\|_2^2
\mid \mathbf Z
\right]
=
O_P\!\left(
n^{-2s/(4s+1)}
\right).
\]
\end{proof}

\subsection{Proof of Theorem~\ref{thm:frsm} (Convergence rate of FRSM)}

\begin{proof}
The FRSM estimator coincides with the functional ridge estimator 
applied to the restricted design matrix \(\mathbf Z_1\) with penalty 
\(\lambda_3 \mathbf R_1\) (see (A8)). 
Hence it has the same algebraic form as the FRE estimator considered 
in Theorem~\ref{thm:fre}, but defined on the relevant block only.

Since \(\boldsymbol{\beta}_1\) satisfies the smoothness condition (A1), 
Lemma~\ref{lem:spline} yields a spline approximation 
\(\boldsymbol{\beta}_1^*\) such that
\[
\|\boldsymbol{\beta}_1^* - \boldsymbol{\beta}_1\|_2^2
=
O(K_z^{-2s}),
\]
where the prediction norm \(\|\cdot\|_2\) is defined via (A6).

Under the penalty conditions in (A2) and the design stability 
assumption (A7), the same bias–variance decomposition 
as in the proof of Theorem~\ref{thm:fre} applies with 
\((\mathbf Z_1,\mathbf R_1,\lambda_3)\) 
in place of \((\mathbf Z,\mathbf R,\lambda_1)\). 
In particular,
\[
E\!\left[
\|\widehat{\boldsymbol{\beta}}_{1,\mathrm{FRSM}}
-
\boldsymbol{\beta}_1
\|_2^2
\mid \mathbf Z_1
\right]
=
O_P\!\left(
K_z^{-2s}
+
\lambda_3^2 K_z^{2(m-s)}
+
\frac{K_z}{\lambda_3 n}
\right),
\]
where the variance term follows from Lemma~\ref{lem:varbound} 
under (A3), (A5), and (A7).

Choosing
\[
K_z \asymp n^{1/(4s+1)},
\qquad
\lambda_3 \asymp n^{-2s/(4s+1)},
\]
balances the approximation and variance terms. 
Since \(m \le s\) by (A2), the penalisation bias term 
is of smaller order, and we conclude
\[
E\!\left[
\|\widehat{\boldsymbol{\beta}}_{1,\mathrm{FRSM}}
-
\boldsymbol{\beta}_1
\|_2^2
\mid \mathbf Z_1
\right]
=
O_P\!\left(
n^{-2s/(4s+1)}
\right).
\]

Finally, if \(\boldsymbol{\beta}_2 = \mathbf 0\), 
the omitted predictors do not affect the regression function, 
so the FRSM estimator remains consistent and achieves 
the same optimal rate.
\end{proof}

\subsection{Proof of Theorem~\ref{thm:frfm} (Convergence rates of FRFM)}

\begin{proof}

We partition the design and penalty matrices as
\[
\mathbf Z = [\mathbf Z_1 , \mathbf Z_2],
\qquad
\mathbf P_\lambda
=
\operatorname{diag}(\lambda_1 \mathbf R_1, \lambda_2 \mathbf R_2),
\]
and write
\(
\boldsymbol{\beta}
=
(\boldsymbol{\beta}_1^\top,\boldsymbol{\beta}_2^\top)^\top.
\)

By Lemma~\ref{lem:spline}, for each block \(j=1,2\) there exists a spline
approximation \(\boldsymbol{\beta}_j^*\) with coefficient vector
\(\mathbf b_j^*\) such that
\[
\|\boldsymbol{\beta}_j^* - \boldsymbol{\beta}_j\|_2
=
O(K_z^{-s}),
\qquad
\|\mathbf b_j^*\|_{\mathbf R_j}^2
=
O(K_z^{2(m-s)}).
\]

\medskip

The FRFM estimator satisfies
\[
\widehat{\mathbf b}
=
(\mathbf Z^\top \mathbf Z + \mathbf P_\lambda)^{-1}
\mathbf Z^\top \mathbf y.
\]
By block matrix inversion,
\[
\widehat{\mathbf b}_1
=
(\mathbf Z_1^\top \mathbf M \mathbf Z_1 + \lambda_1 \mathbf R_1)^{-1}
\mathbf Z_1^\top \mathbf M \mathbf y,
\]
where
\[
\mathbf M
=
\mathbf I
-
\mathbf Z_2
(\mathbf Z_2^\top \mathbf Z_2 + \lambda_2 \mathbf R_2)^{-1}
\mathbf Z_2^\top.
\]

Under Assumption~(A7), the eigenvalues of
\(n^{-1}\mathbf Z_2^\top \mathbf Z_2\) are bounded away from zero and infinity.
If \(\lambda_2 \to \infty\) sufficiently fast, then
\[
\|\mathbf M - \mathbf I\|_{\mathrm{op}}
=
O_P(\lambda_2^{-1}),
\]
so the nuisance block has asymptotically negligible influence on
\(\widehat{\mathbf b}_1\).
Hence the relevant block behaves as a ridge estimator based on
\((\mathbf Z_1,\lambda_1 \mathbf R_1)\).

Applying Theorem~\ref{thm:fre} to this effective design yields
\[
E\!\left[
\|\widehat{\boldsymbol{\beta}}_{1,\mathrm{FRFM}}
-
\boldsymbol{\beta}_1
\|_2^2
\mid \mathbf Z
\right]
=
O_P\!\left(
n^{-2s/(4s+1)}
\right),
\]
for
\(
K_z \asymp n^{1/(4s+1)}
\)
and
\(
\lambda_1 \asymp n^{-2s/(4s+1)}.
\)

\medskip

For the nuisance block, two cases arise.

\smallskip
\noindent
\textit{(i) If } \(\boldsymbol{\beta}_2 = \mathbf 0\).  
Then there is no approximation bias and, by Lemma~\ref{lem:varbound},
\[
E\!\left[
\|\widehat{\boldsymbol{\beta}}_{2,\mathrm{FRFM}}\|_2^2
\mid \mathbf Z
\right]
=
O_P\!\left(
\frac{K_z}{\lambda_2 n}
\right).
\]
With \(\lambda_2 \asymp n\), this is \(O_P(n^{-1})\), which is faster than the
relevant-block rate.

\smallskip
\noindent
\textit{(ii) If } \(\boldsymbol{\beta}_2\) is \(s\)-smooth.  
Using the same bias–variance decomposition as in Theorem~\ref{thm:fre},
\[
E\!\left[
\|\widehat{\boldsymbol{\beta}}_{2,\mathrm{FRFM}}
-
\boldsymbol{\beta}_2
\|_2^2
\mid \mathbf Z
\right]
=
O_P\!\left(
K_z^{-2s}
+
\lambda_2^2 K_z^{2(m-s)}
+
\frac{K_z}{\lambda_2 n}
\right).
\]
Balancing these terms with
\(
\lambda_2 \asymp n^{2s/(4s+1)}
\)
yields the same optimal rate
\[
n^{-2s/(4s+1)}.
\]

\medskip

Combining the rates for both blocks establishes the stated convergence
result for FRFM.
\end{proof}


\section{Proof of Asymptotic Normality}
\label{app:normality_proof}

\begin{proof}

We work under assumptions (A1)–(A9) together with the undersmoothing
conditions (U1)–(U4).  Let $\boldsymbol\beta$ denote the true coefficient vector
and let $\widehat{\boldsymbol\beta}_n$ be its penalised spline estimator obtained from
\[
\widehat{\mathbf b}_n
=
(\mathbf Z^\top \mathbf Z + \mathbf R_\lambda)^{-1}
\mathbf Z^\top \mathbf y,
\qquad
\mathbf y = \mathbf Z \mathbf b_0 + \boldsymbol\varepsilon,
\]
where $\mathbf b_0$ is the coefficient vector of the $L^2$ projection
$\Pi_{K_z}\boldsymbol\beta$ onto the spline space.

By (A1) and standard spline approximation theory,
\[
\|\Pi_{K_z}\boldsymbol\beta - \boldsymbol\beta\|_2
=
O(K_z^{-s}).
\]
Condition (U4) ensures that $K_z^{-s} = o(n^{-1/2})$, so the projection bias
is negligible at the $\sqrt n$ scale.  Consequently, it suffices to analyze
$\widehat{\mathbf b}_n - \mathbf b_0$.

Define $\mathbf G_n = n^{-1}\mathbf Z^\top \mathbf Z$
and $\mathbf M_n = \mathbf G_n + n^{-1}\mathbf R_\lambda$.
A direct decomposition yields
\[
\sqrt n(\widehat{\mathbf b}_n - \mathbf b_0)
=
\mathbf M_n^{-1}
\frac{1}{\sqrt n}\mathbf Z^\top \boldsymbol\varepsilon
-
\sqrt n\, \mathbf M_n^{-1} (n^{-1}\mathbf R_\lambda)\mathbf b_0.
\]

The second term represents the penalty bias.  Under (A2), the eigenvalues of $\mathbf R_\lambda$ scale as $O(\lambda_{j,n}K_z^{2m-1})$, while (A7)
implies $\|\mathbf M_n^{-1}\| = O_P(1)$.
Hence
\[
\left\|
\sqrt n\,\mathbf M_n^{-1}(n^{-1}\mathbf R_\lambda)\mathbf b_0
\right\|
=
O_P\!\bigl(\sqrt n\,\lambda_{j,n} K_z^{2m-1}\bigr),
\]
which converges to zero under (U2)–(U3).
Thus the penalty bias is asymptotically negligible.

To control the leading term, observe that
\[
\mathbf M_n^{-1} - \mathbf G_n^{-1}
=
-\mathbf G_n^{-1}(n^{-1}\mathbf R_\lambda)\mathbf M_n^{-1}.
\]
It follows that
\[
\|\mathbf M_n^{-1} - \mathbf G_n^{-1}\|
=
O_P\!\bigl(n^{-1}\lambda_{j,n}K_z^{2m-1}\bigr).
\]
Multiplying this bound by the score vector
\(
\frac{1}{\sqrt n}\mathbf Z^\top \boldsymbol\varepsilon
\),
whose norm is $O_P(\sqrt{K_z})$ under (A5),
and using (U1)–(U3), we obtain
\[
(\mathbf M_n^{-1} - \mathbf G_n^{-1})
\frac{1}{\sqrt n}\mathbf Z^\top \boldsymbol\varepsilon
=
o_P(1).
\]
Consequently,
\[
\sqrt n(\widehat{\mathbf b}_n - \mathbf b_0)
=
\mathbf G_n^{-1}
\frac{1}{\sqrt n}\mathbf Z^\top \boldsymbol\varepsilon
+
o_P(1).
\]

Under (A7), the eigenvalues of $\mathbf G_n$ are uniformly bounded away
from zero and infinity.  Moreover, by Lemma A.3 of Cardot et al.\ (2003),
\[
\|\mathbf G_n - \mathbf G\|_{\mathrm{op}}
=
O_P(\sqrt{K_z/n}),
\]
where $\mathbf G$ is the matrix representation of the covariance operator
$\Gamma$ in the spline basis.  Since $K_z = o(n^{1/2})$ by (U1),
\[
\|\mathbf G_n^{-1} - \mathbf G^{-1}\|_{\mathrm{op}}
=
O_P(\sqrt{K_z/n})
=
o_P(1).
\]
Therefore,
\[
\sqrt n(\widehat{\mathbf b}_n - \mathbf b_0)
=
\mathbf G^{-1}
\frac{1}{\sqrt n}\sum_{i=1}^n \mathbf z_i \varepsilon_i
+
o_P(1).
\]

Let $x\in L^2(\mathcal T)^p$ be fixed with spline coefficient vector
$\mathbf w_x$.  Since the projection bias is $o(n^{-1/2})$,
\[
\sqrt n(\widehat\Psi_n(x) - \Psi(x))
=
\frac{1}{\sqrt n}\sum_{i=1}^n
\mathbf w_x^\top \mathbf G^{-1}\mathbf z_i \varepsilon_i
+
o_P(1).
\]

Define
\[
\eta_i
=
\mathbf w_x^\top \mathbf G^{-1}\mathbf z_i \varepsilon_i.
\]
By (A5), the $\eta_i$ are independent with mean zero and finite fourth moments.
Their variance is
\[
\mathbb E[\eta_i^2]
=
\sigma^2 \mathbf w_x^\top \mathbf G^{-1} \mathbf w_x.
\]
Hence, by the Lindeberg–Lévy central limit theorem,
\[
\frac{1}{\sqrt n}\sum_{i=1}^n \eta_i
\xrightarrow{d}
\mathcal N
\bigl(
0,
\sigma^2 \mathbf w_x^\top \mathbf G^{-1} \mathbf w_x
\bigr).
\]

It remains to verify consistent variance estimation.
Under (A5) and (U1)–(U3), standard ridge regression arguments yield
\[
\widehat\sigma_n^2 \xrightarrow{\mathbb P} \sigma^2,
\qquad
\mathbf V_n
=
(\mathbf Z^\top\mathbf Z + \mathbf R_\lambda)^{-1}
\mathbf Z^\top\mathbf Z
(\mathbf Z^\top\mathbf Z + \mathbf R_\lambda)^{-1}
=
\mathbf G_n^{-1} + o_P(1),
\]
while $\operatorname{tr}(\mathbf S_n) = O_P(pK_z) = o_P(n)$.
Since $\mathbf w_x^\top \mathbf G^{-1} \mathbf w_x > 0$ by (A9),
Slutsky’s theorem gives
\[
\frac{\widehat\Psi_n(x) - \Psi(x)}{\sqrt{\widehat V_n(x)}}
\xrightarrow{d}
\mathcal N(0,1),
\]
with $\widehat V_n(x) = \widehat\sigma_n^2\,\mathbf w_x^\top\mathbf V_n\mathbf w_x$.
This completes the proof.
\end{proof}

\section{Bias--Variance Decomposition for Ridge Estimators}
\label{app:ridge_decomp}

This appendix derives the exact integrated mean squared error (IMSE) of a generic functional ridge estimator
in the finite‑dimensional spline basis representation, measured in the \(L^2(\mathcal T)\) norm.
The result is standard and is included only for completeness. \\ \\

Let
\[
\widehat{\mathbf b}
=
(\mathbf Z^\top\mathbf Z+\lambda\mathbf R)^{-1}\mathbf Z^\top\mathbf y,
\qquad
\mathbf y=\mathbf Z\mathbf b+\boldsymbol\epsilon,
\]
where $\E[\boldsymbol\epsilon]=\mathbf 0$ and $\V(\boldsymbol\epsilon)=\sigma^2\mathbf I_n$.
Assume $\beta(s)=\boldsymbol\psi(s)^\top\mathbf b$ for a spline basis
$\{\psi_k\}_{k=1}^{K_z}$ spanning a finite-dimensional space $\mathcal S_{K_z}\subset L^2(\mathcal T)$,
and define the Gram matrix
\[
\mathbf G
=
\int_{\mathcal T}\boldsymbol\psi(s)\boldsymbol\psi(s)^\top\,ds,
\]
so that
$
\|\widehat\beta-\beta\|_{L^2}^2
=
\|\widehat{\mathbf b}-\mathbf b\|_{\mathbf G}^2
$
with $\|\mathbf v\|_{\mathbf G}^2=\mathbf v^\top\mathbf G\mathbf v$.

\begin{proposition}[IMSE decomposition]
\label{prop:imse_general}
Under the above setup,
\[
\begin{aligned}
\mathrm{IMSE}(\widehat{\mathbf b})
&:=
\E\!\left[\|\widehat{\mathbf b}-\mathbf b\|_{\mathbf G}^2\right] \\
&=
\bigl\|
\bigl[(\mathbf Z^\top\mathbf Z+\lambda\mathbf R)^{-1}\mathbf Z^\top\mathbf Z-\mathbf I\bigr]\mathbf b
\bigr\|_{\mathbf G}^2
+\sigma^2\,\mathrm{tr}  \!\left(
\mathbf G(\mathbf Z^\top\mathbf Z+\lambda\mathbf R)^{-1}\mathbf Z^\top\mathbf Z
(\mathbf Z^\top\mathbf Z+\lambda\mathbf R)^{-1}
\right).
\end{aligned}
\]
\end{proposition}

\begin{proof}
Set
\[
\mathbf A=(\mathbf Z^\top\mathbf Z+\lambda\mathbf R)^{-1}\mathbf Z^\top,
\qquad
\mathbf H=\mathbf A\mathbf Z=(\mathbf Z^\top\mathbf Z+\lambda\mathbf R)^{-1}\mathbf Z^\top\mathbf Z.
\]
Then
\[
\widehat{\mathbf b}
=
\mathbf A(\mathbf Z\mathbf b+\boldsymbol\epsilon)
=
\mathbf H\mathbf b+\mathbf A\boldsymbol\epsilon,
\qquad
\widehat{\mathbf b}-\mathbf b
=
(\mathbf H-\mathbf I)\mathbf b+\mathbf A\boldsymbol\epsilon.
\]
Using $\|\mathbf v\|_{\mathbf G}^2=\mathbf v^\top\mathbf G\mathbf v$ and expanding,
\[
\E\!\left[\|\widehat{\mathbf b}-\mathbf b\|_{\mathbf G}^2\right]
=
\|(\mathbf H-\mathbf I)\mathbf b\|_{\mathbf G}^2
+2\,\E\!\left[\boldsymbol\epsilon^\top\mathbf A^\top\mathbf G(\mathbf H-\mathbf I)\mathbf b\right]
+\E\!\left[\boldsymbol\epsilon^\top\mathbf A^\top\mathbf G\mathbf A\boldsymbol\epsilon\right].
\]
The cross term is zero because $\E[\boldsymbol\epsilon]=\mathbf 0$.
Moreover, since $\mathbf A^\top\mathbf G\mathbf A$ is symmetric,
\[
\E\!\left[\boldsymbol\epsilon^\top\mathbf A^\top\mathbf G\mathbf A\boldsymbol\epsilon\right]
=
\mathrm{tr}\!\left(\mathbf A^\top\mathbf G\mathbf A\,\V(\boldsymbol\epsilon)\right)
=
\sigma^2\,\mathrm{tr}\left(\mathbf A^\top\mathbf G\mathbf A\right)
=
\sigma^2\,\mathrm{tr}\left(\mathbf G\mathbf A\mathbf A^\top\right),
\]
where the last step uses cyclicity of the trace. Finally,
\[
\mathbf A\mathbf A^\top
=
(\mathbf Z^\top\mathbf Z+\lambda\mathbf R)^{-1}\mathbf Z^\top\mathbf Z
(\mathbf Z^\top\mathbf Z+\lambda\mathbf R)^{-1},
\]
which gives the stated variance term. Combining the two parts yields the claim.
\end{proof}

\begin{remark}
Although Proposition~\ref{prop:imse_general} is formulated for the generic functional ridge estimator (FRE), it provides the structural decomposition underlying the convergence behavior of all three estimators. 
For the FRSM, the same bias–variance trade-off applies on the restricted design. However, if $\boldsymbol{\beta}_2 \neq \mathbf{0}$, the estimator incurs an additional bias arising from the omission of nonzero coefficients in the excluded block.For the FRFM, the use of differential penalization ($\lambda_1 \neq \lambda_2$) induces non-uniform shrinkage across blocks. This allows the relevant component $\widehat{\boldsymbol{\beta}}_1$ to attain the optimal rate while the nuisance component $\widehat{\boldsymbol{\beta}}_2$ is increasingly attenuated as $\lambda_2$ grows. The resulting convergence behavior reflects the block-adaptive structure of the estimator.
\end{remark}

\newpage

\begin{landscape}
	\begin{table}[p]
		\centering
		\caption{\textbf{IMSE values for FRE, FRFM, and FRSM under different sample sizes ($n$), correlation levels ($\rho$), and error variances ($\sigma^2$) for $p=20$.}}
		\label{tab:IMSE_all_n_p20}
		
		\renewcommand{\arraystretch}{2.00}
		\setlength{\tabcolsep}{10pt}
		
		\resizebox{\linewidth}{!}{%
			\begin{tabular}{llccc ccc ccc}
				\hline
				& & \multicolumn{3}{c}{$\sigma^2=0.5$}
				& \multicolumn{3}{c}{$\sigma^2=1.0$}
				& \multicolumn{3}{c}{$\sigma^2=10.0$} \\
				\cline{3-5}\cline{6-8}\cline{9-11}
				$n$ & Model
				& $\rho=0.5$ & $\rho=0.8$ & $\rho=0.99$
				& $\rho=0.5$ & $\rho=0.8$ & $\rho=0.99$
				& $\rho=0.5$ & $\rho=0.8$ & $\rho=0.99$ \\
				\hline
				
				\multirow{3}{*}{25}
				& FRE  & 11.8457 & 7.8509 & 5.3408 & 11.9544 & 7.6832 & 5.6552 & 12.9948 & 8.5898 & 11.3928 \\
				& FRFM & 14.0575 & 9.7509 & 5.4721 & 14.5222 & 9.9860 & 5.4902 & 14.3471 & 10.3437 & 7.6090 \\
				& FRSM & 5.8306 & 4.9845 & 5.0164 & 6.1386 & 5.0165 & 5.4476 & 6.3285 & 5.4490 & 11.3139 \\
				\hline
				
				\multirow{3}{*}{50}
				& FRE  & 14.4963 & 10.2394 & 10.9963 & 14.5682 & 10.5497 & 11.3129 & 15.8303 & 11.0590 & 26.8397 \\
				& FRFM & 7.1173 & 4.2765 & 8.5118 & 7.4738 & 4.5225 & 8.6006 & 8.4176 & 5.5391 & 16.2439 \\
				& FRSM & 3.9229 & 3.8081 & 4.1809 & 3.8537 & 3.8577 & 4.2685 & 4.0956 & 3.9237 & 6.0330 \\
				\hline
				
				\multirow{3}{*}{100}
				& FRE  & 3.4604 & 3.3245 & 4.4191 & 3.4348 & 3.3619 & 4.5896 & 3.5929 & 3.5127 & 6.9624 \\
				& FRFM & 0.2213 & 0.1721 & 2.1756 & 0.2249 & 0.1903 & 2.2547 & 0.4324 & 0.3503 & 3.3225 \\
				& FRSM & 2.9688 & 3.2259 & 3.8497 & 2.9935 & 3.2276 & 3.8962 & 3.0366 & 3.2709 & 4.7371 \\
				\hline
				
				\multicolumn{11}{l}{}
			\end{tabular}
		}
		
	\end{table}
\end{landscape}

\begin{landscape}
	\begin{table}[p]
		\centering
		\caption{\textbf{IMSE values for FRE, FRFM, and FRSM under different sample sizes ($n$), correlation levels ($\rho$), and error variances ($\sigma^2$) for $p=30$.}}
		\label{tab:IMSE_all_n_p30}
		
		\renewcommand{\arraystretch}{2.00}
		\setlength{\tabcolsep}{10pt}
		
		\resizebox{\linewidth}{!}{%
			\begin{tabular}{llccc ccc ccc}
				\hline
				& & \multicolumn{3}{c}{$\sigma^2=0.5$}
				& \multicolumn{3}{c}{$\sigma^2=1.0$}
				& \multicolumn{3}{c}{$\sigma^2=10.0$} \\
				\cline{3-5}\cline{6-8}\cline{9-11}
				$n$ & Model
				& $\rho=0.5$ & $\rho=0.8$ & $\rho=0.99$
				& $\rho=0.5$ & $\rho=0.8$ & $\rho=0.99$
				& $\rho=0.5$ & $\rho=0.8$ & $\rho=0.99$ \\
				\hline
				
				\multirow{3}{*}{25}
				& FRE  & 10.9898 & 8.3352 & 4.4314 & 11.3015 & 8.3615 & 4.5676 & 11.4738 & 8.4687 & 5.9121 \\
				& FRFM & 11.7681 & 9.3760 & 4.2261 & 11.7961 & 9.4624 & 4.2339 & 11.8056 & 9.3228 & 4.7384 \\
				& FRSM & 5.9769 & 4.8995 & 5.1457 & 6.1727 & 5.0064 & 5.4841 & 6.1796 & 5.3415 & 10.7117 \\
				\hline
				
				\multirow{3}{*}{50}
				& FRE  & 20.5252 & 12.1676 & 10.7730 & 21.1643 & 12.2557 & 11.2570 & 21.8487 & 15.2034 & 31.6025 \\
				& FRFM & 23.0834 & 14.3137 & 11.2794 & 23.1704 & 13.8954 & 10.5602 & 23.0631 & 15.3500 & 17.8673 \\
				& FRSM & 3.9097 & 3.7618 & 4.1953 & 3.9025 & 3.8022 & 4.2032 & 4.0648 & 3.9492 & 6.2973 \\
				\hline
				
				\multirow{3}{*}{100}
				& FRE  & 5.6281 & 4.7645 & 5.7804 & 5.8281 & 4.8407 & 5.9708 & 6.0989 & 5.3555 & 11.9878 \\
				& FRFM & 0.7302 & 0.5659 & 3.4944 & 0.8360 & 0.5631 & 3.6448 & 1.2253 & 0.8719 & 6.0269 \\
				& FRSM & 3.0045 & 3.2304 & 3.8421 & 3.0143 & 3.2424 & 3.8865 & 3.0440 & 3.3152 & 4.8042 \\
				\hline
				
				\multicolumn{11}{l}{}
			\end{tabular}
		}
		
	\end{table}
\end{landscape}
\end{document}